\documentclass[acmsmall,nonacm,10pt]{acmart}

\usepackage{booktabs}
\usepackage{subcaption}
\usepackage{wrapfig}

\usepackage{amsmath}
\usepackage{stmaryrd}
\usepackage{mathrsfs}
\usepackage{multirow}

\usepackage{ifthen}
\usepackage{mathpartir}

\usepackage{tikz}
\usepackage{tikz-layers}
\usetikzlibrary{arrows.meta,intersections,automata,shapes,positioning,calc}

\usepackage{fontawesome}

\newcommand{\margus}[1]{\textcolor{red}{\ifmmode \text{[Margus: #1]}\else [Margus: #1] \fi}}
\newcommand{\olli}[1]{\textcolor{blue}{\ifmmode \text{[Olli: #1]}\else [Olli: #1] \fi}}
\newcommand{\tom}[1]{\textcolor{purple}{\ifmmode \text{[Tom: #1]} \else [Tom: #1] \fi}}
\newcommand{\gabriel}[1]{\textcolor{magenta}{\ifmmode \text{[Gabriel: #1]} \else [Gabriel: #1] \fi}}

\newcommand{\ERE}[1][{}]{\mathcal{R}_{#1}}
\newcommand{\RE}[1][{}]{\mathcal{R}^{\textsc{std}}_{#1}}
\newcommand{\eps}{(\hspace{-.12em})}
\newcommand{\inter}{\land}

\newcommand{\union}{\lor}
\newcommand{\compl}{\texttt{\char`~}}
\newcommand{\fuse}{\oplus}
\newcommand{\conc}{\cdot}
\newcommand{\st}{{\ast}}
\newcommand{\plus}{{+}}
\newcommand{\One}[2][{}]{\mathit{One}_{#1}(#2)}
\newcommand{\DFA}[1][{}]{\mathit{DFA}_{#1}}
\newcommand{\NullName}{\mathit{Null}}
\newcommand{\Null}[1]{\NullName(#1)}

\newcommand{\BC}[1]{\mathbb{B}(#1)}
\newcommand{\BCp}[1]{\mathbb{B}^+\!(#1)}
\newcommand{\DNF}[1]{\mathrm{DNF}(#1)}

\newcommand{\A}{\mathcal{A}}

\newcommand{\PA}{\Psi_{\!\A}}
\newcommand{\den}[2][{}]{\llbracket #2 \rrbracket_{#1}}
\newcommand{\andA}{\mathrel{\&}}
\newcommand{\orA}{\mathrel{|}}
\newcommand{\notA}{{\mathrm{!}}}

\newcommand{\D}{\mathbb{D}}
\newcommand{\Do}{\D^\omega}
\newcommand{\Ds}{\D^*}
\newcommand{\Dp}{\D^+}

\newcommand{\Compl}[1]{\complement\!\left(#1\right)}

\newcommand{\bD}{\mathfrak{D}}

\newcommand{\TT}[2][{}]{{\ifthenelse{\equal{#2}{}}{\mathcal{T}_{#1}\,}{\mathcal{T}_{#1}\langle #2\rangle}}}
\newcommand{\Leaves}{\Phi}
\newcommand{\leafof}[2]{#1{\downarrow}#2}
\newcommand{\LeavesOf}[1]{\textit{Leaves}(#1)}
\newcommand{\CondOf}[1]{\textit{Conds}(#1)}

\newcommand{\LangName}{\mathscr{L}}
\newcommand{\Lang}[2][{}]{\ifthenelse{\equal{#1}{}}{\ifthenelse{\equal{#2}{}}{\LangName}{\LangName(#2)}}{\LangName_{#1}(#2)}}

\newcommand{\FLangName}{\mathcal{L}}
\newcommand{\FLang}[2][{}]{\ifthenelse{\equal{#1}{}}{\ifthenelse{\equal{#2}{}}{\FLangName}{\FLangName(#2)}}{\FLangName_{#1}(#2)}}

\newcommand{\DERIV}[2][{}]{\mathbf{D}_{#1}(#2)}

\newcommand{\derName}{\partial}
\newcommand{\der}[1]{\derName(#1)}

\newcommand{\derivName}{\delta}
\newcommand{\deriv}[2][{}]{\ifthenelse{\equal{#2}{}}{{\derivName_{#1}}}{\ifthenelse{\equal{#1}{}}{\derivName(#2)}{D_{#1}^{\textit{Brz}}(#2)}}}

\newcommand{\Qi}[1][{}]{Q^{\textrm{\tiny0}}_{\textrm{\scriptsize$\begin{array}{@{}c@{}}\\[-2.1em]\!#1\\[-1em]\end{array}$}}}

\newcommand{\Minterms}[1]{\textit{Minterms}(#1)}

\newcommand{\eqdef}{\stackrel{\textsc{\tiny def}}{=}}
\newcommand{\IFF}{\Leftrightarrow}
\newcommand{\IMP}{\Rightarrow}

\newcommand{\Nat}{\mathbb{N}}

\newcommand{\True}{\mathbf{true}}
\newcommand{\False}{\mathbf{false}}
\newcommand{\bnot}{\mathbf{not}}
\newcommand{\bor}{\mathrel{\mathbf{or}}}
\newcommand{\band}{\mathrel{\mathbf{and}}}

\newcommand{\ITEBrace}[3]{
\left\{
\begin{array}{@{}l@{\;}l@{}}
#2, & \textrm{\textbf{if} $#1$;} \\
#3, & \textrm{\textbf{otherwise}.}
\end{array}
\right.
}
\newcommand{\IfThenElse}[3]{\textbf{if}\; #1\; \textbf{then}\; #2\; \textbf{else}\; #3}
\newcommand{\IfThen}[2]{\textbf{if}\; #1\; \textbf{then}\; #2}

\newenvironment{myeq}{\\[.2em]\mbox{~}\begin{math}}{\end{math}\\[.2em]}

\newcommand{\Until}{\mathrel{\mathbf{U}}}
\newcommand{\Next}{\mathbf{X}}
\newcommand{\Finally}{\mathbf{F}}
\newcommand{\Globally}{\mathbf{G}}
\newcommand{\Release}{\mathrel{\mathbf{R}}}

\newcommand{\ite}[3]{(#1 \,{\boldsymbol{{?}}}\, #2 \,{\boldsymbol{{:}}}\, #3)}
\newcommand{\ifthen}[2]{(#1 \,{\boldsymbol{{?}}}\, #2)}

\newcommand{\DERS}[1]{\mathrel{\raisebox{-.1em}{\scriptsize$\xrightarrow{\begin{array}{@{}c@{}}#1\\[-.4em]\end{array}}$}}}
\newcommand{\tr}[3]{{#1}\,{\DERS{#2}}\,{#3}}

\newcommand{\LTL}[1][{}]{\textit{LTL}_{#1}}
\newcommand{\LTLp}[1][{}]{\textit{LTL}_{#1}^+}
\newcommand{\RLTL}[1][{}]{\textit{RLTL}_{#1}}
\newcommand{\RLTLp}[1][{}]{\textit{RLTL}_{#1}^+}

\newcommand{\Rest}[2][1]{#2_{(#1..)}}
\newcommand{\ith}[2]{#1_{#2}}
\newcommand{\Prefix}[2]{#2_{(..#1)}}

\newcommand{\M}[1][{}]{\mathcal{M}_{#1}}

\newcommand{\N}[1][{}]{\mathcal{N}_{#1}}

\newcommand{\Pair}[2]{\langle #1, #2\rangle}

\newcommand{\ABA}[1][{}]{\textit{ABW}_{#1}}
\newcommand{\NBA}[1][{}]{\textit{NBW}_{#1}}
\newcommand{\DBA}[1][{}]{\textit{DBW}_{#1}}

\newcommand{\Osat}[1]{O^{\textrm{sat}}_{#1}}

\newcommand{\limplies}{\rightarrow}

\newcommand{\tf}{\varrho}
\newcommand{\tfc}{\rho}

\newcommand{\tfs}[3][{}]{\hat{\tf}_{#1}(#2,#3)}
\newcommand{\ders}[3][{}]{\hat{\derName}_{#1}(#2,#3)}
\newcommand{\derivs}[3][{}]{\hat{\derivName}_{#1}(#2,#3)}

\newcommand{\restr}[2]{#1{\restriction}#2}

\newcommand{\INFname}{\mathit{INF}}
\newcommand{\INF}[1]{\INFname(#1)}

\newcommand{\setset}[1]{\{\{#1\}\}}
\newcommand{\set}[1]{\{#1\}}

\newcommand{\DomainOf}[1]{\textit{Dom}(#1)}

\newcommand{\targetsof}[1]{\mathit{States}(#1)}

\newcommand{\aprod}{\otimes}

\newcommand{\bvarphi}{\boldsymbol{\varphi}}
\newcommand{\bpsi}{\boldsymbol{\psi}}
\newcommand{\bphi}{\boldsymbol{\phi}}
\newcommand{\bQ}{\mathbf{Q}}
\newcommand{\bP}{\mathbf{P}}
\newcommand{\boldp}{\mathbf{p}}
\newcommand{\boldq}{\mathbf{q}}

\newcommand{\Fin}[1]{#1\aprod \{\emptyset\}}

\newcommand{\tfINF}{\boldsymbol{\tf}}

\newcommand{\unaryop}{\textrm{\scriptsize$\blacklozenge$}}

\newcommand{\UNSAT}[1]{\textit{UNSAT}(#1)}

\newcommand{\mt}[1]{\widehat{#1}}

\newcommand{\AElimName}{\textrm{\AE}}
\newcommand{\AElim}[1]{\AElimName\left(#1\right)}

\newcommand{\mytikz}[5]{
\begin{math}
\begin{array}{@{}c@{}}
\\[#1]
\hspace{#3}
#5
\hspace{#4}
\\[#2]
\end{array}
\end{math}
}
\newcommand{\lowerlabel}[2][0em]{\begin{array}{@{}c@{}}#2\\[#1]\end{array}}
\newcommand{\raiselabel}[2][0em]{\begin{array}{@{}c@{}}\\[#1]#2\end{array}}

\newcommand{\pconj}[1]{\pi_{#1}}
\newcommand{\pform}[1]{{\,\widetilde{#1}\,}}
\newcommand{\pset}[1]{\lfloor{#1}\rfloor}
\newcommand{\toconj}[1]{{\curlywedge}#1}

\newcommand{\wcl}[1]{\boldsymbol{\{}#1\boldsymbol{\}}}
\newcommand{\nwcl}[1]{\boldsymbol{!\{}#1\boldsymbol{\}}}
\newcommand{\ocl}[1]{#1^{\omega}}
\newcommand{\eimpl}{\mathrel{{\diamond}\hspace{-.13em}{\shortrightarrow}}}
\newcommand{\uimpl}{\mathrel{\raisebox{.06em}{\scriptsize$\square$}\hspace{-.13em}{\shortrightarrow}}}

\newcommand{\anchorpred}{\$}
\newcommand{\anchorsymb}{\texttt{\scriptsize\char`$}}

\newcommand{\MinSet}[1]{\textit{MinSet}(#1)}

\newcommand{\prece}{\prec_{\epsilon}}

\newcommand{\feq}{\cong}

\newcommand{\Acc}[1]{\textit{Acc}(#1)}

\newcommand{\positive}[1]{(0\,{<}\,#1)}
\newcommand{\divides}[2]{(#1\,{\shortmid}\,#2)}
\newcommand{\ndivides}[2]{(#1\,{\nshortmid}\,#2)}

\makeatletter

\pgfkeys{/pgf/.cd,
  rectangle corner radius/.initial=3pt
}
\newif\ifpgf@rectanglewrc@donecorner@
\def\pgf@rectanglewithroundedcorners@docorner#1#2#3#4{%
  \edef\pgf@marshal{%
    \noexpand\pgfintersectionofpaths
      {%
        \noexpand\pgfpathmoveto{\noexpand\pgfpoint{\the\pgf@xa}{\the\pgf@ya}}%
        \noexpand\pgfpathlineto{\noexpand\pgfpoint{\the\pgf@x}{\the\pgf@y}}%
      }%
      {%
        \noexpand\pgfpathmoveto{\noexpand\pgfpointadd
          {\noexpand\pgfpoint{\the\pgf@xc}{\the\pgf@yc}}%
          {\noexpand\pgfpoint{#1}{#2}}}%
        \noexpand\pgfpatharc{#3}{#4}{\cornerradius}%
      }%
    }%
  \pgf@process{\pgf@marshal\pgfpointintersectionsolution{1}}%
  \pgf@process{\pgftransforminvert\pgfpointtransformed{}}%
  \pgf@rectanglewrc@donecorner@true
}
\pgfdeclareshape{rectangle with rounded corners}
{
  \inheritsavedanchors[from=rectangle] 
  \inheritanchor[from=rectangle]{north}
  \inheritanchor[from=rectangle]{north west}
  \inheritanchor[from=rectangle]{north east}
  \inheritanchor[from=rectangle]{center}
  \inheritanchor[from=rectangle]{west}
  \inheritanchor[from=rectangle]{east}
  \inheritanchor[from=rectangle]{mid}
  \inheritanchor[from=rectangle]{mid west}
  \inheritanchor[from=rectangle]{mid east}
  \inheritanchor[from=rectangle]{base}
  \inheritanchor[from=rectangle]{base west}
  \inheritanchor[from=rectangle]{base east}
  \inheritanchor[from=rectangle]{south}
  \inheritanchor[from=rectangle]{south west}
  \inheritanchor[from=rectangle]{south east}

  \savedmacro\cornerradius{%
    \edef\cornerradius{\pgfkeysvalueof{/pgf/rectangle corner radius}}%
  }

  \backgroundpath{%
    \northeast\advance\pgf@y-\cornerradius\relax
    \pgfpathmoveto{}%
    \pgfpatharc{0}{90}{\cornerradius}%
    \northeast\pgf@ya=\pgf@y\southwest\advance\pgf@x\cornerradius\relax\pgf@y=\pgf@ya
    \pgfpathlineto{}%
    \pgfpatharc{90}{180}{\cornerradius}%
    \southwest\advance\pgf@y\cornerradius\relax
    \pgfpathlineto{}%
    \pgfpatharc{180}{270}{\cornerradius}%
    \northeast\pgf@xa=\pgf@x\advance\pgf@xa-\cornerradius\southwest\pgf@x=\pgf@xa
    \pgfpathlineto{}%
    \pgfpatharc{270}{360}{\cornerradius}%
    \northeast\advance\pgf@y-\cornerradius\relax
    \pgfpathlineto{}%
  }

  \anchor{before north east}{\northeast\advance\pgf@y-\cornerradius}
  \anchor{after north east}{\northeast\advance\pgf@x-\cornerradius}
  \anchor{before north west}{\southwest\pgf@xa=\pgf@x\advance\pgf@xa\cornerradius
    \northeast\pgf@x=\pgf@xa}
  \anchor{after north west}{\northeast\pgf@ya=\pgf@y\advance\pgf@ya-\cornerradius
    \southwest\pgf@y=\pgf@ya}
  \anchor{before south west}{\southwest\advance\pgf@y\cornerradius}
  \anchor{after south west}{\southwest\advance\pgf@x\cornerradius}
  \anchor{before south east}{\northeast\pgf@xa=\pgf@x\advance\pgf@xa-\cornerradius
    \southwest\pgf@x=\pgf@xa}
  \anchor{after south east}{\southwest\pgf@ya=\pgf@y\advance\pgf@ya\cornerradius
    \northeast\pgf@y=\pgf@ya}

  \anchorborder{%
    \pgf@xb=\pgf@x
    \pgf@yb=\pgf@y%
    \southwest%
    \pgf@xa=\pgf@x
    \pgf@ya=\pgf@y%
    \northeast%
    \advance\pgf@x by-\pgf@xa%
    \advance\pgf@y by-\pgf@ya%
    \pgf@xc=.5\pgf@x
    \pgf@yc=.5\pgf@y%
    \advance\pgf@xa by\pgf@xc
    \advance\pgf@ya by\pgf@yc%
    \edef\pgf@marshal{%
      \noexpand\pgfpointborderrectangle
      {\noexpand\pgfqpoint{\the\pgf@xb}{\the\pgf@yb}}
      {\noexpand\pgfqpoint{\the\pgf@xc}{\the\pgf@yc}}%
    }%
    \pgf@process{\pgf@marshal}%
    \advance\pgf@x by\pgf@xa%
    \advance\pgf@y by\pgf@ya%
    \pgfextract@process\borderpoint{}%
    \pgf@rectanglewrc@donecorner@false
    \southwest\pgf@xc=\pgf@x\pgf@yc=\pgf@y
    \advance\pgf@xc\cornerradius\relax\advance\pgf@yc\cornerradius\relax 
    \borderpoint
    \ifdim\pgf@x<\pgf@xc\relax\ifdim\pgf@y<\pgf@yc\relax
      \pgf@rectanglewithroundedcorners@docorner{-\cornerradius}{0pt}{180}{270}%
    \fi\fi
    \ifpgf@rectanglewrc@donecorner@\else
      \southwest\pgf@yc=\pgf@y\relax\northeast\pgf@xc=\pgf@x\relax
      \advance\pgf@xc-\cornerradius\relax\advance\pgf@yc\cornerradius\relax
      \borderpoint
      \ifdim\pgf@x>\pgf@xc\relax\ifdim\pgf@y<\pgf@yc\relax
       \pgf@rectanglewithroundedcorners@docorner{0pt}{-\cornerradius}{270}{360}%
      \fi\fi
    \fi
    \ifpgf@rectanglewrc@donecorner@\else
      \northeast\pgf@xc=\pgf@x\relax\pgf@yc=\pgf@y\relax
      \advance\pgf@xc-\cornerradius\relax\advance\pgf@yc-\cornerradius\relax
      \borderpoint
      \ifdim\pgf@x>\pgf@xc\relax\ifdim\pgf@y>\pgf@yc\relax
       \pgf@rectanglewithroundedcorners@docorner{\cornerradius}{0pt}{0}{90}%
      \fi\fi
    \fi
    \ifpgf@rectanglewrc@donecorner@\else
      \northeast\pgf@yc=\pgf@y\relax\southwest\pgf@xc=\pgf@x\relax
      \advance\pgf@xc\cornerradius\relax\advance\pgf@yc-\cornerradius\relax
      \borderpoint
      \ifdim\pgf@x<\pgf@xc\relax\ifdim\pgf@y>\pgf@yc\relax
       \pgf@rectanglewithroundedcorners@docorner{0pt}{\cornerradius}{90}{180}%
      \fi\fi
    \fi
  }
}

\makeatother

\tikzset{
    buchi/.style={
        >=To,
        initial text=,
        state/.style={rectangle with rounded corners, draw, inner sep=0.1cm, minimum size=10pt},
        boolop/.style={circle, draw, inner sep=0.01cm, minimum size=5pt},
        cond/.style={diamond,aspect=2,draw,inner sep=0.02cm},
        accept/.style={postaction={overlay,double,draw}},
        stem/.style={-},
        move/.style={-To,shorten >=0.02cm},
        macc/.style={-To,shorten >=0.04cm}, 
        acc/.style={shorten >=0.04cm}, 
        movecond/.style={-To},
        btrue/.style={Circle-To,shorten >=0.02cm,shorten <=-0.07cm},
        bfalse/.style={{Circle[black,fill=white]}-To,fill=white,shorten >=0.02cm,shorten <=-0.07cm},
        btruec/.style={Circle-To,shorten <=-0.07cm},
        bfalsec/.style={{Circle[black,fill=white]}-To,fill=white,shorten <=-0.07cm},
        node distance=1cm,
    },
}

\newenvironment{ex}{\begin{example}}{\hfill$\Box$\end{example}}
\newtheorem{thm}{Theorem}
\newtheorem{cor}{Corollary}
\newtheorem{lma}{Lemma}

\begin{document}

\title{Symbolic Automata: $\omega$-Regularity Modulo Theories}

\author{Margus Veanes}
 \orcid{0009-0008-8427-7977}             
 \affiliation{
   \institution{Microsoft Research}            
   \country{USA}                    
 }
 \email{margus@microsoft.com}          

 \author{Thomas Ball}
  \affiliation{
  \institution{Microsoft Research}            
   \country{USA}                    
  }
 \email{tball@microsoft.com}          

\author{Gabriel Ebner}
 \affiliation{
   \institution{Microsoft Research}            
   \country{USA}                    
 }
 \email{gabrielebner@microsoft.com}          

\author{Olli Saarikivi}
\orcid{0000-0001-7596-4734}             
\affiliation{
  \institution{Microsoft Research}            
   \country{USA}                    
 }
 \email{olsaarik@microsoft.com}          

\begin{abstract}

Symbolic automata are finite state automata that support potentially infinite alphabets,
such as the set of rational numbers, generally applied to regular expressions/languages 
over finite words. 
In symbolic automata (or automata modulo $\A$), an alphabet is
represented by an effective Boolean algebra $\A$, supported by a decision procedure
for satisfiability.
Regular languages over infinite words (so called $\omega$-regular languages) 
have a rich history paralleling that of regular languages over finite words,
with well known applications to model checking via B\"uchi automata and 
temporal logics.

We generalize symbolic automata to support $\omega$-regular 
languages via \emph{symbolic transition terms} and \emph{symbolic derivatives},
bringing together a variety of classic automata and logics in a unified framework
that provides all the necessary ingredients to support symbolic model checking
modulo $\A$. In particular, we define: (1) alternating B\"uchi automata modulo $\A$
 as well (non-alternating) non-deterministic B\"uchi automata modulo $\A$ ($\NBA[\A]$);
  (2) an alternation elimination algorithm that incrementally
  constructs an $\NBA[\A]$ from an $\ABA[\A]$, and can also be used for constructing
  the product of two $\NBA[\A]$;
  (3) a definition of linear temporal logic (LTL) modulo $\A$ that generalizes Vardi's
  construction of alternating B\"uchi automata from LTL, using (2) to go from LTL modulo $\A$ to $\NBA[\A]$
  via $\ABA[\A]$.

  Finally, we present a combination of LTL modulo $\A$ with 
  extended regular expressions modulo $\A$ that generalizes
  the Property Specification Language (PSL).
  Our combination 
  allows regex \emph{complement}, that is not supported in PSL
  but can be supported naturally by using symbolic transition terms.

\end{abstract}

\maketitle

\section{Introduction}
\label{sec:intro}

Classical finite automata, which correspond to regular expressions/languages,
are finite in \emph{three orthogonal dimensions}:
\begin{itemize}
  \item the number of \emph{states} in the automata, which make it a \emph{finite} (or \emph{finite-state}) automata;
  \item the number of \emph{alphabet elements} that induce state transitions;
  \item the length of the \emph{words} (a sequence of alphabet elements) recognized by the automata.
\end{itemize}
Symbolic automata \cite{DV21} are finite state automata that support potentially infinite alphabets, 
such as the set of rational numbers. The transitions of symbolic automata are labeled with predicates 
from an effective Boolean algebra~$\A$. Each atomic predicate and their Boolean combination in $\A$
can compactly represent a (possibly infinite) subset of elements.

The power of symbolic automata 
is to make use of decision procedures for $\A$ \cite{BM08} to avoid the 
need for reduction to the propositional (classic) case, which can incur an exponential blowup.
For example, a propositional encoding of a Boolean expression of size $N$ in $\A$ would
require $O(2^N)$ conjunctions of atomic predicates in a propositional setting (so called
``mintermization''). In some cases, a classical algorithm for finite
automata can be lifted straightforwardly to symbolic automata; in other cases,
entirely new algorithms are required \cite{DV21}.

Regular languages over infinite words (so called $\omega$-regular languages) 
have a rich history paralleling that of regular languages over finite words;
In fact, the two are intimately connected: \cite{McNaughton66}
showed that $\omega$-regular languages are equivalently described
by B\"uchi automata and the union of a finite set of $\omega$-regular expressions,
where each expression has the form ${\alpha_i}^* {\beta_i}^\omega$, 
where $\alpha_i$ and $\beta_i$ are regular expressions for finite words
and $\beta_i$ is infinitely repeated using the $\omega$ operator.
Alternating B\"uchi automata have equivalent expressive power as B\"uchi automata, 
but can be exponentially more compact \cite{MH84}. 

Furthermore, $\omega$-regular languages are an important part of the theoretical foundation 
for model checking of reactive systems \cite{CGP99}, where both the system and (many of) the 
properties to be verified of the system are described by sets of infinite traces. Various
forms of automata, including B\"uchi automata, can be used to describe the system while
temporal logics describe the properties to be checked. In particular,
Linear temporal logic (LTL) \cite{Pnueli77} is used to describe the 
expected behavior of reactive systems and
can be lowered to B\"uchi automata via a variety of algorithms.
As LTL does not fully capture $\omega$-regular languages, 
various extensions of LTL incorporate regular expressions \cite{PSL,spotRef} 
to increase its expressive capability. The LTL-based automata-theoretic approach to model checking
constructs the product of the system automata with the B\"uchi automata representing
the negation of the LTL formula and then analyzes whether the resulting automata has any 
accepting runs \cite{VardiW86}. If there are no accepting runs then the system automata 
is verified with respect to the LTL formula. 

\subsection{Symbolic transition terms and symbolic derivatives modulo $\A$}

We generalize symbolic automata to $\omega$-regular 
languages by first lifting the concept of \emph{transition regexes} 
from \cite{StanfordVB21} to \emph{symbolic transition terms}, which are parameteric over two domains:
\begin{itemize}
  \item \emph{alphabet}, with an effective Boolean algebra $\A$ over it;
  \item \emph{leaves}, which represent the language of the automata, with its own effective Boolean algebra
\end{itemize}
Symbolic transition terms can be used as the basis of defining the semantics of languages/logics and their corresponding
automata using \emph{symbolic derivatives}, providing a unifying framework for automated reasoning and rewriting that operates 
incrementally over a symbolic representation (Section~\ref{sec:tt}).

In our formulation, $\tf(\phi)$ is a
\textbf{\emph{symbolic derivative modulo $\A$}} that works as a
\emph{curried} form of a classical transition function $\tfc(\phi,a)$, 
where $a$ is a concrete alphabet element. Note that the decision of actually
evaluating $\tfc(\phi,a)$ is avoided in $\tf(\phi)$ (each predicate in $\A$ represents
a potentially infinite set of alphabet elements).

An example of an LTL formula modulo an SMT solver is
$\varphi\land\psi$ with $\varphi=\Globally\positive{x}$ and
$\psi=\divides{2}{x}\Until\divides{3}{x}$ where $x$ is a variable of
type integer. The formula states that $x$ is always positive and
remains even until divisible by $3$.  Then, for example,
$(x{\mapsto}2)(x{\mapsto}4)(x{\mapsto}3)\ocl{(x{\mapsto}1)}\models\varphi$.
Here $(x{\mapsto}n)$ is a member of the \emph{infinite alphabet}
consisting of all possible SMT interpretations.  The symbolic
derivative of $\varphi\land\psi$ is the following nested
if-then-else (ITE) expression, or \emph{transition term},
\[
\tf(\varphi\land\psi) =
\ite{\positive{x}}{\ite{\divides{3}{x}}{\varphi}{\ite{\divides{2}{x}}{(\varphi\land\psi)}{\bot}}}{\bot}
\]
whose \emph{conditions} are the SMT predicates $\positive{x}$, $\divides{2}{x}$, and $\divides{3}{x}$,
and whose \emph{leaves} are the LTL formulas $\varphi$, $\varphi\land\psi$, and $\bot$ (\textit{false}).
E.g., ${\tf(\varphi\land\psi)}{(x{\mapsto}8)} = \varphi\land\psi$ because
$\positive{8}$,
$\ndivides{3}{8}$, and 
$\divides{2}{8}$.

A very useful \emph{duality} principle of symbolic derivatives
is that $\tf(\lnot\phi) =\lnot\tf(\phi)$. So, by using the
 standard negation law of ITEs
$\lnot\ite{\alpha}{f}{g}=\ite{\alpha}{\lnot f}{\lnot g}$, it follows 
in this particular case that 
\begin{eqnarray*}
\tf(\lnot(\varphi\land\psi)) &=&
\lnot(\ite{\positive{x}}{\ite{\divides{3}{x}}{\varphi}{\ite{\divides{2}{x}}{(\varphi\land\psi)}{\bot}}}{\bot}) \\
&=&
\ite{\positive{x}}{\ite{\divides{3}{x}}{\lnot\varphi}
{\ite{\divides{2}{x}}{\lnot(\varphi\land\psi)}{\top}}}{\top}
\end{eqnarray*}
The same principle avoids the need for \emph{determinization} of \emph{symbolic finite automata}
in the context of \emph{extended regular expression (regex) constraint solving},
by \emph{lazily propagating} complement of regex
derivatives in Z3 \cite{StanfordVB21}.  It also avoids \emph{symbolic
alternating finite automata
(s-AFA) normalization} \cite{DAntoniKW16} needed
for complement, whose transitions are
represented in flat form by $\tr{q}{\alpha}{\phi}$ where $q$ is a
state, $\alpha$ is a predicate of $\A$, and $\phi$ is a positive
Boolean combination of states.  Currently the only known technique to normalize an s-AFA amounts
to computing all the satisfiable Boolean combinations of the guards that is
exponential time in the worst case.

\subsection{Contributions and Overview}

We show how symbolic transition terms and derivatives can be used to generalize symbolic
automata to work with $\omega$-regular languages modulo an infinite alphabet represented by $\A$,
bringing together a variety of classic automata and logics in a unified framework. In particular, we give:
\begin{itemize}
  \item a definition of alternating B\"uchi automata modulo $\A$ 
  ($\ABA[\A]$) and their relationship to classical alternating B\"uchi automata as well as to (non-alternating) non-deterministic B\"uchi automata modulo $\A$ ($\NBA[\A]$) and deterministic B\"uchi automata modulo $\A$ ($\DBA[\A]$) (Section~\ref{sec:ABA})
  \item an alternation elimination algorithm, a symbolic generalization
  of \cite{MH84} that incrementally
  constructs an $\NBA[\A]$ from an $\ABA[\A]$ (Section~\ref{sec:alt-elim})
  \item a definition of linear temporal logic (LTL) \cite{Pnueli77} modulo $\A$ that generalizes Vardi's
  construction of alternating B\"uchi automata from LTL \cite{Vardi95LTL} -- with the previous results, this gives a way to go from LTL modulo $\A$ to $\NBA[\A]$ (Section~\ref{sec:ltl});
\end{itemize}
Finally, we bring the above results together to demonstrate a
combination of LTL modulo $\A$ with extended regular expressions modulo $\A$ 
($\RLTL[\A]$) that generalizes SPOT \cite{spotRef} and PSL \cite{PSL} (Section~\ref{sec:ltlere}). 
Our combination  allows regex \emph{complement}, that is not supported in SPOT/PSL
but can be supported naturally by using transition terms.
We also lift the classical concept of $\omega$-regular
languages \cite{Buchi60,McNaughton66} as the languages accepted by $\ABA$ so as to
be modulo $\A$, and show that $\RLTL[\A]$ precisely captures $\omega$-regularity modulo $\A$.

Section~\ref{sec:related} reviews related work 
and Section~\ref{sec:future} discusses future work. 
Section~\ref{sec:conclude} concludes the paper.
The appendix has proofs for the theorems and lemmas 
whose proofs were not included in the main part of the paper. 

\section{Preliminaries}
\label{sec:prelims}

As a meta-notation throughout the paper
we write $\textit{lhs}\eqdef\textit{rhs}$
to let \textit{lhs} be \emph{equal by definition} to \textit{rhs}.
In the following let $\D$ be a nonempty (possibly infinite) domain of
core elements. 
The natural numbers are denoted by $\Nat$.

\subsection{Effective Boolean Algebras and Boolean Closures}
\label{sec:EBA}

An \emph{effective Boolean algebra over $\D$}
is a tuple $\A=(\D, \Psi, \den{\_}, \bot, \top, \lor, \land, \neg)$
where $\Psi$ is a set of
\emph{predicates} closed under the Boolean connectives; $\den{\_} :
\Psi \rightarrow 2^{\D}$ is a \emph{denotation function}; $\bot, \top
\in \Psi$; $\den{\bot} = \emptyset$, $\den{\top} = \D$, and for all
$\alpha, \beta \in \Psi$, $\den{\alpha \lor \beta} = \den{\alpha}
\cup \den{\beta}$, $\den{\alpha \land \beta} = \den{\alpha} \cap
\den{\beta}$, and $\den{\neg \alpha} = \D \setminus \den{\alpha}$.
We also write $a\vDash\alpha$ for $a\in\den{\alpha}$.
For $\alpha,\beta\in\Psi$ we
write $\alpha\equiv\beta$ to mean $\den{\alpha}=\den{\beta}$.  In
particular, if $\alpha\equiv\bot$ then $\alpha$ is
\emph{unsatisfiable}, also denoted by $\UNSAT{\alpha}$.
We require that the connectives are computable
and that satisfiability is decidable in $\A$.
We use $\A$ as a subscript to indicate a component of $\A$, e.g.,
$\PA$ is the set of predicates of $\A$.  We omit the
subscript when it follows from the context.

Let $\Gamma\subseteq\PA$ be finite.  A \emph{minterm} of $\Gamma$ is
a \emph{satisfiable} predicate $(\bigwedge\!S)\land\lnot\!\bigvee(\Gamma{\setminus} S)$
for some $S\subseteq\Gamma$.  $\Minterms{\Gamma}$, say $\Sigma$,
denotes the set of all minterms of $\Gamma$. The \emph{core
properties} of $\Sigma$ are that all minterms are satisfiable and
mutually disjoint, and that each satisfiable predicate in $\Gamma$ is
equivalent to a disjunction of some minterms. Thus, $\Sigma$ defines
a \emph{finite partition} of $\D$.  For example, if
$\Gamma=\{\alpha,\beta\}$ then $\Sigma$ is the set of all satisfiable
predicates in
$\{\alpha\land\beta,\lnot\alpha\land\beta,\alpha\land\lnot\beta,\lnot\alpha\land\lnot\beta\}$.

We let $\Osat{\A}(n)$ denote the computational complexity of checking
satisfiability in $\A$ for predicates $\psi$ of size $|\psi|=n$.
Here we make the standard assumption that the size of a predicate is
the sum of the sizes of its immediate subexpressions.  Under this assumption it
follows that the complexity of computing $\Sigma$ is
$O(2^{{\Osat{\A}(n)}})$ where $n=\sum_{i<k}|\gamma_i|$ and
$\Gamma=\{\gamma_i\}_{i<k}$.  Observe also that $|\Sigma|\leq 2^{k}$, i.e.,
the \emph{number} of minterms of $\Gamma$ is in the worst-case exponential in $k$.

\paragraph*{Boolean Closure}
\label{sec:BC}

Given a set $S$,
we define the Boolean closure $\BC{S}$ of $S$ such that $S\cup\{\bot,\top\}\subseteq\BC{S}$ and
$\BC{S}$ is closed under $\lor$, $\land$, and $\lnot$, where
$\bot$ and $\top$ act as
the corresponding unit and zero of the Boolean
operations, satisfying $\bot\land s=\bot$, $\top\land s=s$,
$\top\lor s=\top$, $\bot\lor s=s$, $\neg\bot=\top$, and $\neg\top=\bot$.%
\footnote{$S$ may already include $\bot$ and $\top$ and it may also be that $S=\BC{S}$.}
We write $\BCp{S}$ for the subset of $\BC{S}$
where $\lnot$ is omitted.

The disjunctive normal form of $\phi\in\BCp{S}$
is a disjunction $\bigvee_{i=0}^n\psi_i$ where each $\psi_i$ is a conjunction $\bigwedge X_i$
for some $X_i\subseteq S$.
\emph{Formally}, we adopt the \emph{set-of-sets} representation
of $\DNF{\phi}$ as $\{X_i\}_{i=0}^n$.
For example, $\DNF{(s_1\lor s_2)\land s_3}= \{\{s_1,s_3\},\{s_2,s_3\}\}$
that stands for a formula $(s_1\land s_3)\lor(s_2\land s_3)$.
Observe, as a special case, in
DNF, $\emptyset$ denotes $\bot$ and $\{\emptyset\}$ denotes
$\top$. 
One advantage of this representation is that then
commutativity, associativity and idempotence of the \emph{intended semantics} of the Boolean operations
is built-in to $\DNF{\psi}$.

\subsection{Infinite/Finite Sequences and Derivatives}
\label{sec:inf}

An \emph{infinite word} $w \in \Do$ is a function from $\Nat$ to
$\D$. We let $\ith{w}{i}\eqdef w(i)$ denote the $i$'th
element of $w$ for $i\in\Nat$.
If $a\in \D$ and $v\in \Do$ then $u=a{\cdot}v$ denotes the
sequence such that if $i=0$ then $\ith{u}{i}=a$ else $\ith{u}{i} = \ith{v}{i-1}$.  If
$S\subseteq \D$ and $L\subseteq \Do$ then
$S{\cdot}L\eqdef\{a{\cdot} v\mid a\in S, v\in L\}$.  If
$v\in\Do$ and $n\in\Nat$ then $\Rest[n]{v}\eqdef\lambda
i.\ith{v}{n+i}$ denotes
the suffix of $v$ starting with $\ith{v}{n}$.  If $a\in\D$ then $a^\omega\in\Do$ such that
$\ith{a^\omega}{i}=a$ for all $i\in\Nat$. Concatenation ($\cdot$) binds stronger than
intersection ($\cap$) that binds stronger than union ($\cup$).
We use juxtaposition for
concatenation when this is unambiguous.

For \emph{finite words} $v\in\Ds$
we share the same notation as for infinite words that $\ith{v}{i}$ is the
$i$'th element of $v$ and $\Rest[i]{v}$ is the $i$'th rest of $v$, when
$0\leq i < |v|$ with $|v|$ as the length of $v$.
We also let $\Rest[|v|]{v}\eqdef\epsilon$.
Note that  $\Rest[0]{v}=v$.
We also define $x\prec y$ to mean that
$x\in\Ds$ is a nonempty proper prefix of $y$ where $y\in\Ds$ or $y\in\Do$.
We also let $x\prece y \eqdef x\prec y \bor (x=\epsilon\band y\neq\epsilon)$.
We write $vw$ for the concatenation of $v$ with $w$ where
$w\in\Ds$ or $w\in\Do$ and lift this definition to sets, as ususal.

Let $\bD$ be either $\Do$ or $\Ds$, with the
associated definition of complement $\Compl{L}\eqdef\bD\setminus L$ for $L\subseteq\bD$.
The \emph{derivative of $L$ for $a\in\D$} is
$
\DERIV[a]{L} \eqdef \{v\mid av\in L\}
$.
These facts hold for all $a\in\D,L,R\subseteq\bD$:
\[
\DERIV[a]{\Compl{L}} = \Compl{\DERIV[a]{L}},\quad
\DERIV[a]{L\cup R} = \DERIV[a]{L}\cup \DERIV[a]{R},\quad
\DERIV[a]{L\cap R} = \DERIV[a]{L}\cap \DERIV[a]{R}.
\]
\emph{As a convention throughout the paper when a term $t$ is associated
with a language over $\Ds$ we denote that by $\FLang{t}$
and when $t$ is associated with a language 
over $\Do$ we denote that by $\Lang{t}$.}


\section{Transition Terms}
\label{sec:tt}
As shown in the introduction, {\em transition terms} are the target language that 
represents the derivative-based compositional semantics of the source language (such as
regular expressions, LTL, as well as their combination in RLTL); 
they generalize the notion of transition regexes from \cite{StanfordVB21}. 

Let $\A=(\D, \Psi, \den{\_}, \bot, \top, \orA, \andA, \notA)$ be a given effective Boolean algebra
as the given \emph{element algebra} and let $\Leaves$ be a generic type
of \emph{leaves} that in general is independent of $\A$.
\emph{Transition terms} $\TT[\A]{\Leaves}$
($\TT{\Leaves}$ or $\TT{}$ for short when both parameters are clear) are defined as
the following nested if-then-else expressions, i.e., ordered labelled binary
trees whose internal nodes have labels in $\PA$ and leaves belong to $\Leaves$.
We let $\phi$ range over $\Leaves$ and $\alpha$ range over $\PA$.
\begin{itemize}
\item if $\phi\in\Leaves$ then $\phi\in\TT{}$, $\phi$ is called a \emph{leaf};
\item if $\alpha\in\PA$ and $f,g\in\TT{}$ then $\ite{\alpha}{f}{g}\in\TT{}$ is called an
\emph{ITE} and $\alpha$ is called its \emph{condition}.
\end{itemize}
Given $a\in\D$ and $f\in\TT{}$, the \emph{leaf of
$f$ for $a$}, denoted by $\leafof{f}{a}$, is defined as follows.
\[
\leafof{\phi}{a} \eqdef \phi, \qquad 
\leafof{\ite{\alpha}{f}{g}}{a} \eqdef \ITEBrace{a\in\den{\alpha}}{\leafof{f}{a}}{\leafof{g}{a}}
\]
The set of all leaves and conditions of a transition term are defined as follows:
\[
\begin{array}{@{}r@{\;}c@{\;}l@{\quad}r@{\;}c@{\;}l@{}}
\LeavesOf{\phi}&\eqdef&\{\phi\} &
\LeavesOf{\ite{\alpha}{f}{g}}& \eqdef& \LeavesOf{f}\cup\LeavesOf{g}\\
\CondOf{\phi}&\eqdef&\emptyset &
\CondOf{\ite{\alpha}{f}{g}} &\eqdef& \{\alpha\}\cup\CondOf{f}\cup\CondOf{g}
\end{array}
\]
Any binary operation $\diamond:\Leaves\times\Leaves\rightarrow\Leaves'$,
is lifted to $\TT{\Leaves}\times\TT{\Leaves}\rightarrow\TT{\Leaves'}$ as
follows:\footnote{The choice of rule in the case when $h$ is an ITE can be
based on a total ordering over conditions ($\PA$) to maintain uniformity.}
\begin{equation}
\label{eq:binary}
\begin{array}{rcl@{\qquad}rcl}
\ite{\alpha}{f}{g}\diamond h &\eqdef& \ite{\alpha}{f\diamond h}{g\diamond h} &
h\diamond\ite{\alpha}{f}{g} &\eqdef& \ite{\alpha}{h\diamond f}{h\diamond g}
\end{array}
\end{equation}
Any unary operation $\unaryop:\Leaves\rightarrow\Leaves'$ is lifted
to $\TT{\Leaves}\rightarrow\TT{\Leaves'}$ analogously as follows:
\begin{equation}
\label{eq:unary}
\unaryop\ite{\alpha}{f}{g} \eqdef \ite{\alpha}{\unaryop f}{\unaryop g}
\end{equation}
The following generic properties of transition terms are used frequently
in many constructions where $\diamond$ and $\unaryop$ are lifted binary and unary operations.
\begin{lma}
\label{lma:TT}
Let $f,g\in\TT[\A]{\Phi}$ and $a\in\D$. Then
$\leafof{(f\diamond g)}{a}= \leafof{f}{a}\diamond \leafof{g}{a}$
and $\leafof{(\unaryop f)}{a}=\unaryop(\leafof{f}{a})$.
\end{lma}
We say that two transition terms $f$ and $g$
are \emph{functionally equivalent}, that we denote by $f\feq g$, when
$
\forall a\in\D: \leafof{f}{a}=\leafof{g}{a}
$.
We use \emph{condition elimination} $\ite{\alpha}{f}{f}\feq f$ and
\emph{flattening} $\ite{\alpha}{\ite{\beta}{f}{g}}{g}\feq \ite{\alpha\andA\beta}{f}{g}$
(analogously for the symmetrical cases), as particular generic simplifications.

\subsection{Cleaning Modulo $\A$}

The internal branching structure of nested ITEs is irrelevant as long
as the leaf values remain the same for any given element $a\in\D$.
All simplification rules that preserve functional equivalence are
allowed. The most critical rule, that we call \emph{cleaning} here,
eliminates dead branches and involves the use of Boolean operations
and \emph{satisfiability modulo $\A$}.

A transition term $f$ is called \emph{clean} when
all of its subterms are reachable, i.e., when $f$ has
no \emph{infeasible paths}. The \emph{restriction of $f$ to
$\beta\in\PA$} is denoted here by $\restr{f}{\beta}$
and \emph{cleaning of $f$} means to rewrite $f$ to $\restr{f}{\top}$, i.e.,
$f\feq \restr{f}{\top}$.
\[
\restr{\phi}{\beta}\eqdef\phi\qquad
\restr{\ite{\alpha}{f}{g}}{\beta} \eqdef
\left\{
\begin{array}{ll}
\restr{g}{\beta}, &\textrm{if}\; \UNSAT{\beta{\andA}\alpha}; \\
\restr{f}{\beta}, &\textrm{else if}\; \UNSAT{\beta{\andA}\notA\alpha}; \\
\ite{\alpha}{\restr{f}{(\beta{\andA}\alpha)}}{\restr{g}{(\beta{\andA}\notA\alpha)}}, &\textrm{otherwise.}
\end{array}
\right.
\]
For example, ${\ite{\top}{f}{\_}} \feq f$.  All binary operations
(\ref{eq:binary}) can be assumed to produce a clean result (in most uses)
in an implementation, that can take into account that the arguments
are clean to begin with, when this is known, and also use caching to
avoid redundant cleaning operations.

\begin{ex}
Let $\alpha,\beta\in\PA$ be two predicates such that
$\emptyset\neq\den{\alpha}\subsetneq\den{\beta}\neq\D$.
Then
\[
\begin{array}{@{}r@{\;}c@{\;}l@{}}
\restr{(\ite{\beta}{\__1}{\__2}\land\ite{\alpha}{\__3}{\__4})}{\top} &=&
\restr{\ite{\beta}{\__1\land\ite{\alpha}{\__3}{\__4}}{\__2\land\ite{\alpha}{\__3}{\__4}}}{\top}
\\
&=& \restr{\ite{\beta}{\ite{\alpha}{\__1\land\__3}{\__1\land\__4}}{\ite{\alpha}{\__2\land\__3}{\__2\land\__4}}}{\top}
\\
&=& \ite{\beta}{\restr{\ite{\alpha}{\__1\land\__3}{\__1\land\__4}}{\beta}}{
                 \restr{\ite{\alpha}{\__2\land\__3}{\__2\land\__4}}{\notA\beta}}
                 \\
&=& \ite{\beta}{\ite{\alpha}{\restr{(\__1\land\__3)}{(\beta{\andA}\alpha)}}{\restr{(\__1\land\__4)}{(\beta{\andA}\notA\alpha)}}}{
                 \restr{(\__2\land\__4)}{\notA\beta}}
\end{array}
\]
where cleaning eliminated the infeasible path to $\__2\land\__3$
because $\notA\beta{\andA}\alpha$ is unsatisfiable.
\end{ex}

\subsection{Boolean Operations and INF Normal Form}

We will be considering $\Leaves=\BCp{S}$ where $S$ is some set
whose members are called \emph{states} of an automaton,
or $S$ is a set of formulas and $S=\BC{S}$.
In either case (\ref{eq:binary}) and (\ref{eq:unary}) are
used to lift the Boolean operations to $\TT{\Leaves}$.

When we work with $\phi\in\Leaves=\BCp{S}$ it is useful to work instead with
the canonical representation $\bphi = \DNF{\phi}$. Recall, that
$\bphi \subseteq 2^S$ represents the formula
$\bigvee_{X\in \bphi}\bigwedge X$ in $\Leaves$.  In particular, $\bphi
= \emptyset$ represents $\bot$ and $\bphi = \{\emptyset\}$ represents
$\top$ (the empty conjunction is, by definition, $\top$, in $\bigvee_{X\in \{\emptyset\}}\bigwedge X$).

We say that $f\in\TT{\DNF{\Leaves}}$ is the \emph{ITE-DNF} or
\emph{INF normal form} of the corresponding
term in $\TT{\Leaves}$. The following definitions of
$\land$ and $\lor$ preserve DNF, and are lifted to $\TT{\DNF{\Leaves}}$ by (\ref{eq:binary}).
Let $\varphi,\psi\in\Leaves$, $\boldsymbol{\varphi}=\DNF{\varphi}$ and $\boldsymbol{\psi}=\DNF{\psi}$.
Then 
\[
\begin{array}{rcl}
\boldsymbol{\varphi}\lor \boldsymbol{\psi} &\eqdef&
\boldsymbol{\varphi}\cup \boldsymbol{\psi} = \DNF{\varphi\lor\psi}\\
\boldsymbol{\varphi}\land \boldsymbol{\psi} &\eqdef&
\{X\cup Y\mid X\in \boldsymbol{\varphi},Y\in \boldsymbol{\psi}\} = \DNF{\varphi\land\psi}
\end{array}
\]
One advantage of INF is that if $S$ is \emph{finite} then so is
$\DNF{\BCp{S}}$.
A potential disadvantage is the worst-case exponential complexity of
computing DNF from a given formula. However, in practice INF will be
applied both \emph{lazily} (on a ``need to know'' basis) and \emph{locally} to individual leaves
that tend to be small, rather than globally to large formulas.

When working with $f$ and $g$ in $\TT{\Leaves}$ we lift their functional equivalence to hold modulo INF:
\[
f \feq g \quad\eqdef\quad \forall a\in\D: \leafof{\INF{f}}{a} = \leafof{\INF{g}}{a}
\]
In particular observe that $\bot\land f \feq \bot$, $\top\land f \feq f$,
$\bot\lor f \feq f$, and $\bot \land f \feq \bot$.
We use the convention that when we write $\textit{lhs}\feq\textit{rhs}$ then
it implicitly means that
\textit{rhs} is a \emph{simplification} of \textit{lhs}.
Observe also that $\leafof{\INF{f}}{a} = \DNF{\leafof{f}{a}}$.

\emph{We frequently use the abbreviation $\ifthen{\alpha}{f}\eqdef\ite{\alpha}{f}{\bot}$
as well as the following functional equivalences in $\TT{\Leaves}$:
$\ifthen{\alpha}{f}{\land}\ifthen{\beta}{g}\feq
\ifthen{\alpha{\andA}\beta}{f{\land}g}$ and
$\ifthen{\alpha}{f}\lor\ifthen{\beta}{f}\feq
\ifthen{\alpha{\orA}\beta}{f}$.}

\section{Alternating B\"uchi Automata Modulo $\A$}
\label{sec:ABA}

We now turn our attention to 
alternating B\"uchi automata (which admit a linear time  
translation from LTL \cite{TsayV21}) and show how they can be generalized
using transition terms so as to be modulo
$\A=(\D, \Psi, \den{\_}, \bot, \top, \orA, \andA, \notA)$.
The key aspect of the following definition is that transitions
in the automata are represented symbolically by 
transition terms in $\TT[\A]{\BCp{Q}}$.

An {\em alternating B\"uchi automaton modulo $\A$} ($\ABA[\A]$) is
a tuple $\M = (\A, Q, \Qi, \tf, F)$ where $Q$ is a \emph{finite}
set of \emph{states}, $\Qi\in \BCp{Q}$ is an \emph{initial state combination}, 
$F\subseteq Q$ is a set of \emph{accepting states}, and
$\tf:Q\rightarrow{\TT[\A]{\BCp{Q}}}$ is a \emph{transition function}
such that if $\top$ occurs in $\tf$ then
$\top\in F$ and $\tf(\top)=\top$, and if
$\bot$ occurs in $\tf$ then
$\bot\in Q{\setminus}F$ and $\tf(\bot)=\bot$.

\begin{wrapfigure}{@{}c@{}}{4cm}
\mytikz{-1.8em}{-1em}{-3cm}{-0.4cm}{
    \begin{tikzpicture}[buchi]
      \small
      \begin{scope}[on above layer]
        \coordinate (init) at (0,0);
        \coordinate (andX) at (1.3cm,0);
        \coordinate (center) at (2cm,0);
          \node[boolop,above=0.3cm of andX] (and2) {\textrm{\scriptsize$\land$}};
          \node[boolop,below=0.3cm of andX] (and1) {\textrm{\scriptsize$\land$}};
          \node[state,accept,right=0.2cm of init] (q0) {$q_0$};
          \node[state,below=0.3cm of center] (q1) {$q_1$};
          \node[state,above=0.3cm of center] (q2) {$q_2$};
          \node[state,accept,right=1cm of center] (top) {$\top$};   
      \end{scope}
      \draw[move] (init) -- (q0);
      \draw[move] (q0) to[out=-45,in=180] node[below] {$\notA\alpha$} (and1);
      \draw[move] (and1) -- (q1);
      \draw[macc] (and1) -- (q0);
      \draw[move] (q1) edge[loop right] node[right] {$\!\notA\alpha$} (q1);
      \draw[move] (q0) to[out=45,in=180] node[above] {$\alpha$} (and2);
      \draw[move] (and2) -- (q2);
      \draw[macc] (and2) -- (q0);
      \draw[move] (q2) edge[loop right] node[right] {$\!\alpha$} (q2);
      \draw[macc] (q2) -- node[left] {$\!\notA\alpha$} (top);
      \draw[macc] (q1) -- node[left] {$\alpha$} (top);
      \draw[move] (top) edge[loop above] node[right] {$\!\top$} (top);
    \end{tikzpicture}
}
\caption{$\ABA$ for $\Globally(\Finally\alpha \land \Finally\notA\alpha)$.\label{fig:aba}}
\end{wrapfigure}
$\M$ is \emph{nondeterministic} or an $\NBA[\A]$ if
$\land$ occurs neither in $\Qi$ nor in $\tf$; $\M$
is \emph{deterministic} or a $\DBA[\A]$ if $\Qi\in Q$ and
neither $\land$ nor $\lor$ occur
in $\tf$.  We indicate a component of $\M$ by using $\M$ as the
subscript, e.g., $\tf_{\M}$, unless $\M$ is clear from the context.

In figures we show each transition term $\tf(q)$ by \emph{symbolic
transitions} $\tr{q}{\beta}{\phi}$ where $\phi\in\BCp{Q}$ is a leaf of $\tf(q)$
\emph{guarded} by the accumulated branch condition $\beta\in\PA$ from the root
of $\tf(q)$ to $\phi$.\footnote{If there are multiple occurrences of
$\phi$ in $\tf(q)$ then $\beta$ is the disjunction in $\A$ of the
corresponding branch conditions. We typically omit transitions to $\bot$
and transitions into disjunctions of states
in an $\NBA$ are shown as separate transitions.}

For example, the $\ABA[\A]$ in
Figure~\ref{fig:aba} has $\Qi=q_0$ and
\[
\tf = 
\{q_0{\mapsto}\ite{\alpha}{q_2 {\land} q_0}{q_1 {\land} q_0},\;
  q_1{\mapsto}\ite{\alpha}{\top}{q_1},\;
  q_2{\mapsto}\ite{!\alpha}{\top}{q_2},\;
  \top{\mapsto}\top\}
\]
Boolean operations $\diamond\in\{\land,\lor\}$ over $\M,\N\in\ABA[\A]$,
such that $Q_{\M}\cap Q_{\N} \subseteq \{\bot,\top\}$,
are defined by:
\[
\M\diamond\N\eqdef (\A,\; Q_{\M}\cup Q_{\N},\; \Qi[\M]\diamond\Qi[\N],\; \tf_{\M}\cup\tf_{\N},\; F_{\M}\cup F_{\N})
\]
We say that $\M$ (or $\tf_{\M}$) is \emph{clean} when all
the transition terms in $\tf_{\M}$ are clean. For example,
the $\ABA[\A]$ in Figure~\ref{fig:aba} is clean provided that
both $\alpha$ and $\notA\alpha$ are satisfiable in $\A$,
i.e., $\alpha\nequiv\bot\nequiv\notA\alpha$.

\subsection{Runs and Languages}
Let $\M = (\A, Q, \Qi, \tf, F)$ be an $\ABA$.
The standard defintion of runs uses $X\subseteq Q$
to specify \emph{satisfying assignments} for $\varphi\in\BCp{Q}$.%
\footnote{
$X\subseteq Q$
\emph{satisfies} $\phi\in\BCp{Q}$ iff the valuation
$\{q\mapsto\True\mid q\in X\}\cup\{q\mapsto\False\mid q\in Q\setminus
X\}$ satisfies $\phi$ \cite{Vardi95LTL}.}
Here we make use of the
fact that $X$ satisifies $\varphi$ iff some $X_{\min} \subseteq X$ satisfies $\varphi$
\emph{in a minimal manner} \cite[5.1]{Kupf18}, defined by $X_{\min}$ being an element of:%
\footnote{Recall that $\DNF{\phi} \subseteq 2^Q$ where $\DNF{\bot}=\emptyset$ and $\DNF{\top}=\set{\emptyset}$.}
\begin{equation}
\label{eq:MinSat}
\{Z \mid Z \in \DNF{\phi}\band
\nexists Y\in \DNF{\phi}:Y\subsetneq Z\}
\end{equation}
Since $(\ref{eq:MinSat})\subseteq\DNF{\varphi}$,
it suffices below to limit the satisfying assignments $X$ of $\varphi$
to $X\in\DNF{\varphi}$.

Let $aw\in \Do$
and $q\in Q$.  A \emph{run for $aw$ from
$q$} is a (potentially \emph{infinite}) $Q$-labelled tree $\tau$ whose root node has label
$q$ and there exists $X\in\DNF{\leafof{\tf(q)}{a}}$ such that,
for each $p\in X$,
there is an immediate subtree $\tau_p$ of $\tau$ such that $\tau_p$ is a run for $w$ from $p$.

A run $\tau$ is \emph{accepting} if \emph{all} of
the infinite branches of $\tau$ visit $F$ infinitely often,
where an \emph{infinite branch} visits the root of $\tau$ and then continues
as an infinite branch from some immediate subtree of $\tau$.
For $w\in
\Do$ and $q\in Q$: \emph{$w$ is accepted from $q$} iff there exists an
accepting run for $w$ from $q$; the \emph{language of $q$ in $\M$} is
defined and lifted to $\BCp{Q}$ as follows:
\begin{equation}
\label{eq:LABAq}
\begin{array}{l}
\Lang[\M]{q} \eqdef
\{w \in\Do\mid \textrm{$w$ is accepted from $q$}\}
\\
\Lang[\M]{\boldq{\lor}\boldp} \eqdef \Lang[\M]{\boldq}\cup\Lang[\M]{\boldp}
\quad
\Lang[\M]{\boldq{\land}\boldp} \eqdef \Lang[\M]{\boldq}\cap\Lang[\M]{\boldp}
\quad
\Lang{\M} \eqdef\Lang[\M]{\Qi}
\end{array}
\end{equation}
It follows that
$\Lang{\M\land\N} = \Lang{\M}\cap\Lang{\N}$ and
$\Lang{\M\lor\N} = \Lang{\M}\cup\Lang{\N}$, corresponding to the
classical case \cite[Theorem~20]{Kupf18}.

Let the \emph{INF normal form} of $\tf$ be
$\tfINF(q) \eqdef \INF{\tf(q)}$ for all $q\in Q$. Thus, $\tfINF$ is
semantically equivalent to $\tf$ and it follows from the
definition of INF that
$\leafof{\tfINF(q)}{a}=\DNF{\leafof{\tf(q)}{a}}$. E.g., in the
$\ABA[\A]$ in Figure~\ref{fig:aba},
$\tfINF(q_0)=\ite{\alpha}{\setset{q_2,q_0}}{\setset{q_1,q_0}}$.

The notion of a language of $\M$ \emph{relative} to a state $q\in Q_{\M}$ is
typically left implicit classically.
Recall that a classical $\ABA$ is a tuple $M = (\Sigma, Q, \Qi, \tfc, F)$
where $\tfc:Q\times\Sigma\rightarrow\BCp{Q}$ where we assume,
{w.l.o.g}, that $\tfc$ is \emph{total} for $\Sigma$ and $\Qi\in\BCp{Q}$.
For
$q\in Q$, let $M^{(q)}\eqdef (\Sigma, Q, q, \tfc, F)$ where $q$ is
thus the \emph{acting initial state}.  We use the definition
$\Lang[M]{q}\eqdef\Lang{M^{(q)}}$ below.

\subsection{From  $\ABA$ Modulo $\A$ to Classical $\ABA$}
\label{sec:toClassicalABA}

A variety of decidability results for $\ABA$ modulo $\A$ 
follow directly from the fact that any $\M$, as defined above,
can be reduced to a classical $\ABA$ (this does not imply a practical
implementation scheme, however, as it requires mintermization). 

Let $\Sigma=\Minterms{\Gamma}$ where
$\Gamma=\CondOf{\M}\eqdef\CondOf{\tf}\eqdef\bigcup_{q\in Q}\CondOf{\tf(q)}$,
and for $a\in\D$ let $\mt{a}\in\Sigma$
denote the minterm such that $a\in\den{\mt{a}}$.  It follows that
\[
\forall a,b\in\D,q\in Q: \mt{a}=\mt{b} \IMP \leafof{\tf(q)}{a}=\leafof{\tf(q)}{b}
\]
from minterms because for any ITE
$\ite{\gamma}{f}{g}$ in $\tf$ we have that $\gamma\in\Gamma$, and thus,
$
\den{\mt{a}}\cap\den{\gamma}\neq\emptyset\IFF \den{\mt{a}}\subseteq\den{\gamma}
$.
Moreover, this holds \emph{for all
Boolean combinations $\gamma$ over $\Gamma$}, i.e., individual elements of
a minterm are indistinguishable in $\M$.

Thus, the \emph{classical transition function}
$\mt{\tf}(q,\mt{a})\eqdef\leafof{\tf(q)}{a}$
for $q\in Q$ and $\mt{a}\in\Sigma$
is well-defined in the $\ABA$ $\mt{\M}\eqdef(\Sigma,Q,\Qi,\mt{\tf},F)$. We lift
$\mt{a}\in\Sigma$ to $\mt{w}\in\Sigma^\omega$.

\begin{lma}[Minterm]
\label{lma:mt}
$\forall\M\in\ABA[\A],\phi\in \BCp{Q_{\M}},w\in\Do:\mt{w}\in\Lang[\mt{\M}]{\phi} \IFF
w\in \Lang[\M]{\phi}.$
\end{lma}

Lemma~\ref{lma:mt} is a useful tool in lifting results from
the classical setting.  In
particular, \emph{decidability of nonemptiness} of $\M$, provided
that $\A$ is effective, follows then directly from decidability of
nonemptiness of $\mt{\M}$, by reduction to $\NBA$ \cite{MH84} and
decidability of nonemptiness of $\NBA$ \cite{Rabin70}.
The cost of computing $\Sigma$
can be exponential in $|\M|$ (the size of $\M$) that is essentially
determined by the size of $\tf$, observe also that $|\Sigma| \leq
2^{|\Gamma|}$, and may involve exponentially many Boolean operations
and satisfiability checks in $\A$.

\subsection{Nondeterministic and Clean Case}

Decidability of nonemptiness of a clean $\NBA[\A]$ has \emph{linear
time} complexity, that follows by adapting the corresponding result
from \cite{EL85a,EL85b}, where cleaning, as a separate preprocessing
step, has time complexity $\Osat{\A}(|\M|)$. In particular, we
implemented the
\emph{nested depth-first cycle detection algorithm} \cite[Algorithm~B]{CVWY92}
without any real changes for nonemptiness checking of a clean
$\NBA[\A]$.  Many other decision problems, such as \emph{incremental
dead-state detection} \cite{SV23}, can be used similarly. The up-front
worst-case exponential cost of converting $\M$ to a classical $\ABA$ $\mt{\M}$
is often avoidable in the symbolic setting by lifting of the
classical algorithms to work modulo $\A$. Similar observations 
hold in the case of automata modulo $\A$ over finite words \cite{DV21}.

\subsection{From Classical $\ABA$ to $\ABA$ Modulo $\A$}
\label{sec:fromClassicalABA}

Here we show how a classical $\ABA$ $M=(\Sigma,Q,\Qi,\tfc,F)$
is lifted to $M_{\A}=(\A,Q,\Qi,\tf,F)$. Let $\A$ be a Boolean
algebra over $\Sigma$ {s.t.} for all $a\in\Sigma$ there is
$\hat{a}\in\PA$ {s.t.} $\den{\hat{a}}=\{a\}$, and
\[
\begin{array}{rcl}
  \forall q\in Q:\tf(q)&\eqdef& \bigvee_{a\in\Sigma}\ifthen{\hat{a}}{\tfc(q,a)}
\end{array}
\]
It
follows that for all $q\in Q$ and $a\in\Sigma$,
$\leafof{\tf(q)}{a}=\tfc(q,a)$ and thus that $M$ and
$M_{\A}$ are equivalent.
A conceptually simple algebra $\A$ is the
powerset algebra $(\Sigma, 2^\Sigma, \lambda x.x, \emptyset, \Sigma,
\cup, \cap, \lambda x.\Sigma{\setminus} x)$ over $\Sigma$ where $\hat{a}=\{a\}$.
A more interesting algebra $\A$ can be constructed when
$\Sigma\subseteq 2^P$ where $P$ is a nonempty finite set of atomic
propositions.  Let $\A$ be a SAT solver over $P$.
Then
$\hat{a}\eqdef({\Large\andA}_{p\in a}\, p)\andA({\Large\andA}_{p\in P\setminus
    a}\,\notA p)$ and thus $\den{\hat{a}}=\{a\}$,
    i.e., $\hat{a}$ is intuitively the ``minterm'' corresponding to $a$.
Note that $\tf$ can also be condensed via rewrites.
For example,
$\ifthen{\hat{a}}{f}\lor\ifthen{\hat{b}}{f} \feq
\ifthen{\hat{a}{\orA}\hat{b}}{f}$.


\section{Alternation Elimination}
\label{sec:alt-elim}

Although alternating B\"uchi automata are an attractive target for
temporal logic, incurring no space blowup, the problem of 
testing for nonemptiness is harder for alternating automata
(due to the presence of both conjunction and disjunction in the state
formula) than non-alternating automata \cite{BokerKR10}.

Let $\M = (\A, Q, \Qi, \tf, F)$ be an $\ABA[\A]$.
We introduce an alternation elimination algorithm for $\M$ that
is a symbolic derivative-based generalization of the algorithm in \cite{MH84}
(see also \cite[Proposition~20]{Vardi95LTL}).
The algorithm constructs an $\NBA[\A]$ $\N = (\A, S, S_0, \sigma, F_{\N})$
accepting the same language as $\M$.
As shown by \cite{BokerKR10}, the concept behind  \cite{MH84} accurately captures
the prime concern in alternation removal, which is the need to associate the 
states of the equivalent $\NBA[\A]$ $\N$ with two sets of states from $\M$, which means that 
a $\Omega(3^{|Q|})$ space blowup cannot be avoided.

In the rest of this section let $\bQ=2^{2^Q}$, $P=2^Q\times 2^Q$, and $\bP=2^P$.
The set $\bQ$ is the type of leaves of transition terms of $\M$ in INF,
i.e., $\tfINF: Q\rightarrow \TT{\bQ}$.
The set $\bP$ is the type of leaves of transition terms of $\N$
where $P\supseteq S$ is the \emph{state type} of $\N$, i.e.,
$\sigma: S\rightarrow\TT{\bP}$.
$F_{\N}=\{\Pair{U}{V}\in S\mid U=\emptyset\}$.%
\footnote{Here we can replace $\bot$ and $\top$ with some states $q_{\bot}\in Q{\setminus}F$ and $q_{\top}\in F$, with $\tf(q_{\bot})=q_{\bot}$ and $\tf(q_{\top})=q_{\top}$,
or equivalently, DNF will eliminate $\bot$ and $\top$.}

\subsection{Algorithm \AElimName}

We start by defining the key operation of \emph{alternation product}
of two transition terms $f,g\in\TT{\bQ}$ as a transition term $f\aprod
g\in\TT{\bP}$.  Let $\bvarphi,\bpsi\in\bQ$ and we use (\ref{eq:binary}) to lift $\aprod$ to
$\TT{\bQ}\times\TT{\bQ}\rightarrow\TT{\bP}$.
\begin{equation}
\label{eq:product}
\begin{array}{@{}r@{\;}c@{\;}l@{\quad}r@{\;}c@{\;}l@{}}
\bvarphi\aprod\bpsi &\eqdef& \{\Pair{X\setminus F}{Y\cup(X\cap F)}\mid X\in\bvarphi,Y\in\bpsi\}
\end{array}
\end{equation}
The purpose of this construction is to partition the states into
those that have not yet visited a final state ($X\setminus F$) and those that have
($Y\cup(X\cap F)$).
For $h\in\TT{\bP}$
let $\targetsof{h}\eqdef\bigcup\LeavesOf{h}$.
Observe that $\targetsof{h} \subseteq P$ as it
is the union of all the leaves of $h$ where each leaf
is a set of states.

The algorithm {\AElimName} constructs an equivalent $\NBA[\A]$ $\AElim{\M}$ from $\M$.
Let $\DomainOf{\sigma}$ denote the set of
all states in $S$ for which the partial transition
function $\sigma$ is currently defined in the while loop.  Upon
termination $\DomainOf{\sigma} = S$.
For $X\subseteq Q$,
let
$
\tfINF(X) \eqdef \IfThenElse{X=\emptyset}{\set{\emptyset}}{\bigwedge_{q\in X}\tfINF(q)}
$
in:
\[
\AElim{\M}\eqdef
\left\{
\begin{array}{l}
S_0 \leftarrow \Fin{\DNF{\Qi}}; S \leftarrow S_0; \sigma \leftarrow \emptyset;
\\
\textbf{while}\; \exists\,\Pair{U}{V}\in S\setminus\DomainOf{\sigma}\; \textbf{do}
\\
  \begin{array}{r@{\;}l}
    \textbf{let}& f =
\left\{
\begin{array}{@{}l@{\;}l@{}}
\tfINF(U)\aprod\tfINF(V), & \textrm{\textbf{if} $U\neq\emptyset$} \\
\Fin{\tfINF(V)}, & \textrm{\textbf{if} $U=\emptyset$}
\end{array}
\right\}
  \end{array}
      \; \textbf{in}\;\;
\sigma(\Pair{U}{V})\leftarrow f;\; S \leftarrow S \cup \targetsof{f}
\\
\textbf{return}\; (\A, S, S_0, \sigma, \{\Pair{U}{V}\in S\mid U=\emptyset\})
\end{array}
\right.
\]

We can make use of the following additional \emph{state reduction} lemma in the algorithm.
The effect of the lemma is already visible in Example~\ref{ex:alt-elim-sample}
where we consider a sample run of the algorithm. Without state reduction,
the number of states would roughly \emph{double} in Example~\ref{ex:alt-elim-sample}.

For $\Pair{U'}{V'},\Pair{U}{V}\in P$ let
$\Pair{U'}{V'}\subseteq\Pair{U}{V}$ $\eqdef$ 
$U'\subseteq U$ and $V'\subseteq V$ and
if $U\neq\emptyset$ then $U'\neq\emptyset$.

\begin{lma}[State Reduction]
\label{lma:state-reduction}
If $\Pair{U'}{V'}\subseteq\Pair{U}{V}$, $\tfINF(U')=\tfINF(U)$, and
$\tfINF(V')=\tfINF(V)$, 
then $\Pair{U}{V}$ can be replaced by $\Pair{U'}{V'}$ in the algorithm.
\end{lma}
\begin{proof}
Accepting condition is preserved.  The computation of $f$ preserves
language equivalence using the fact that if we identify states that
have functionally equivalent transition terms and the same acceptance
condition then the language does not change.
\end{proof}

\begin{ex}
\label{ex:alt-elim-sample}
Let $\M$ be the $\ABA[\A]$ in Figure~\ref{fig:aba}, where $\Qi=q_0$, $F=\{q_0,\top\}$, and
\[
\tf = 
\{q_0\mapsto\ite{\alpha}{q_2 \land q_0}{q_1 \land q_0},\quad
  q_1 \mapsto\ite{\alpha}{\top}{q_1},\quad
  q_2 \mapsto\ite{!\alpha}{\top}{q_2},\quad
  \top\mapsto\top\}
\]
Each row of the following table represents
one iteration of while-loop in the algorithm.
Observe that $\Fin{\DNF{q_0}} = \Fin{\setset{q_0}} = \set{\Pair{\emptyset}{\set{q_0}}}$
because $q_0\in F$.
\[
\begin{array}{@{}l|l|l@{}}
s & \sigma(s) & \AElim{\M} \\ \hline
s_0 = \Pair{\emptyset}{\set{q_0}} &
\Fin{\ite{\alpha}{\setset{q_2,q_0}}{\setset{q_1,q_0}}} &
\multirow{4}{*}{
\textrm{
\mytikz{-1em}{-.5em}{-5em}{-0.8em}{
    \begin{tikzpicture}[buchi]
      \small
      \begin{scope}[on above layer]
        \coordinate (init) at (0,0);
        \coordinate (center) at (1.5cm,0);
          \node[state,accept,right=0.2cm of init] (s0) {$s_0$};
          \node[state,below=0.5cm of center] (s1) {$s_1$};
          \node[state,above=0.5cm of center] (s2) {$s_2$};
      \end{scope}
      \draw[move] (init) -- (s0);
      \draw[move] (s0) -- node[above] {$\notA\alpha$} (s1);
      \draw[move] (s1) edge[loop right] node[right] {$\notA\alpha$} (s1);
      \draw[move] (s0) -- node[left] {$\alpha$} (s2);
      \draw[move] (s2) edge[loop right] node[right] {$\alpha$} (s2);
      \draw[macc] (s2) to[out=180,in=90] node[above] {$\notA\alpha$} (s0);
      \draw[macc] (s1) to[out=180,in=-90] node[below] {$\alpha$} (s0);
    \end{tikzpicture}}}
}
\\
&
= \ite{\alpha}{\set{ \Pair{\set{q_2}}{\set{q_0}}}}{\set{\Pair{\set{q_1}}{\set{q_0}}}}
&
\\
s_1 = \Pair{\set{q_1}}{\set{q_0}} &
\tfINF(q_1){\aprod}\tfINF( q_0)
= \ite{\alpha}{\set{\Pair{\emptyset}{\underbrace{\set{q_2,q_0}}_{\raiselabel[-1.7em]{{\leadsto}\set{q_0}}}}}}
              {\set{\Pair{\set{q_1}}{\underbrace{\set{q_1,q_0}}_{\raiselabel[-1.7em]{{\leadsto}\set{q_0}}}}}} &
\\
s_2 = \Pair{\set{q_2}}{\set{q_0}} &
\tfINF(q_2){\aprod}\tfINF( q_0) =  \ite{\notA\alpha}{\set{\Pair{\emptyset}{\set{q_0}}}}{\set{\Pair{\set{q_2}}{\set{q_0}}}} &
\end{array}
\]
The computation of $\sigma(s_1)$ (and $\sigma(s_2)$)
applied Lemma~\ref{lma:state-reduction} (indicated by $\leadsto$)
by using that $\tfINF(q_2) \land \tfINF(q_0) = \tfINF(q_1) \land  \tfINF(q_0) = \tfINF(q_0)$.
The algorithm terminates with the $\NBA$ shown in the last column.
\end{ex}

\subsection{Correctness of $\AElimName$}
The correctness proof of the algorithm is by reduction to \cite{MH84},
using generic properties of transition terms and mintermization. In
particular, Theorem~\ref{thm:AE} implies that
$\Lang{\AElim{\M}}=\Lang{\M}$.  First note that the algorithm terminates
because $P$ is finite and $S\subseteq P$.

Let $\toconj{}$ below denote the
operation that given a pair $\Pair{U}{V}$ of finite
sets $U,V$ of states, forms their combined conjunction:
$\toconj{\Pair{U}{V}}\eqdef \bigwedge(U\cup V)$.
Note that $\bigwedge\!\emptyset = \top$.

\begin{thm}[\AElimName]
\label{thm:AE}
$\forall q\in Q_{\AElim{\M}}$: $\Lang[\AElim{\M}]{q} = \Lang[\M]{\toconj{q}}$.
\end{thm}
\begin{proof}
Let $\N=\AElim{\M}$ as above.
We use the following properties:
\begin{itemize}
\item[(i)]
the order of members of a conjunction is immaterial due to sets in INF;
\item[(ii)]
by definition of INF and Lemma~\ref{lma:TT}:
$\forall f:Q\rightarrow\TT{\BCp{Q}},a\in\D:\leafof{\INF{f}}{a}=\DNF{\leafof{f}{a}}$.
\end{itemize}
Let $q=\Pair{U}{V}$ be such that $U\neq\emptyset$. We get that
(note that $V=\emptyset$ implies that $\tfINF(V)=\{\emptyset\}$)
\begin{eqnarray}
\nonumber
\leafof{\sigma(\Pair{U}{V})}{a} &=& \leafof{(\tfINF(U)\aprod\tfINF(V))}{a}
\;\stackrel{\textrm{(i)}}{=}\;
\leafof{(\INF{\bigwedge_{t\in U}\tf(t)}\aprod\INF{\bigwedge_{t\in V}\tf(t)})}{a} \\
\nonumber
&\stackrel{\textrm{(ii)}}{=}&
\DNF{\bigwedge_{t\in U}(\leafof{\tf(t)}{a})}\aprod \DNF{\bigwedge_{t\in V}(\leafof{\tf(t)}{a})}\\
\label{eq:prop1}
&\stackrel{(\ref{eq:product})}{=}& \{\Pair{X{\setminus}F}{Y\cup(X\cap F)}\mid
X\in \DNF{\bigwedge_{t\in U}(\leafof{\tf(t)}{a})},
Y\in \DNF{\bigwedge_{t\in V}(\leafof{\tf(t)}{a})}
\end{eqnarray}
Now let $q=\Pair{U}{V}$ be such that $U=\emptyset$.
We get that
\begin{eqnarray}
\nonumber
\leafof{\sigma(\Pair{U}{V})}{a} &=& \leafof{(\Fin{\tfINF(V)})}{a} 
\;\stackrel{\textrm{(i)}}{=}\; \leafof{(\Fin{\INF{\bigwedge_{t\in V}\tf(t)}})}{a}
\;\stackrel{\textrm{(ii)}}{=}\; \Fin{\DNF{\bigwedge_{t\in V}(\leafof{\tf(t)}{a})}} \\
\label{eq:prop2}
&\stackrel{(\ref{eq:product})}{=}&
\{\Pair{X{\setminus}F}{X \cap F}\mid X \in \DNF{\bigwedge_{t\in V}(\leafof{\tf(t)}{a})}\}
\end{eqnarray}
Finally, we relate the properties (\ref{eq:prop1}) and (\ref{eq:prop2})
with $\mt{\M}$ (recall Section~\ref{sec:toClassicalABA}), where for $W\subseteq Q$,
$\bigwedge_{t\in W}(\leafof{\tf(t)}{a})= \bigwedge_{t\in W}\mt{\tf}(t,\mt{a})$.
The properties (\ref{eq:prop1}) and (\ref{eq:prop2})
represent the construction used in \cite{MH84}
on top of which further analysis over accepting runs of $\mt{\M}$
is built.
Assume, {w.l.o.g.}, that $\CondOf{\sigma}=\CondOf{\tf}$.
It follows, for all $s\in S$ and $w\in\Do$, that
\[
\begin{array}{c}
w \in \Lang[\N]{s} \stackrel{\textrm{Lma~\ref{lma:mt}}}{\IFF}
\mt{w}\in  \Lang[\mt{\N}]{s}
\stackrel{\textrm{MH84}}{\IFF}
\mt{w}\in \Lang[\mt{\M}]{\toconj{s}}
\stackrel{\textrm{Lma~\ref{lma:mt}}}{\IFF}
w\in \Lang[\M]{\toconj{s}}
\end{array}
\]
via mintermization and
MH84 \cite{MH84}, and where
we used that minimal satisfiers (\ref{eq:MinSat}) of $\phi$ are
included in $\DNF{\phi}$ in conditions (\ref{eq:prop1}) and (\ref{eq:prop2}).
\end{proof}
Observe that \emph{cleaning} of transition terms
constructed by $\aprod$
(and that $\A$ is \emph{decidable})
 is not necessary for correctness
but essential for any practical implementation of
the algorithm.

\subsection{Product of $\NBA$s Modulo $\A$ as a Special Case of $\AElimName$}

An immediate application of the alternation elimination algorithm
is to construct the \emph{product} of two $\NBA[\A]$'s 
$\N[i]=(\A,Q_i,Q^0_i,\tf_i,F_i)$, for $i\in\{1,2\}$, 
such that $Q_1\cap Q_2\subseteq\{\bot,\top\}$, as the $\NBA[\A]$:
\[
\begin{array}{l}
\N[1]\times\N[2] \eqdef \AElim{\N[1]\land\N[2]}
\end{array}
\]
Thus, the algorithm provides an essential step in the automata-based approach to 
model checking, as discussed in the introduction. 
We make some observations of the
properties of {\AElimName} in this case that result
in the following corollary of Theorem~\ref{thm:AE}.
This also implies that the total number of states in $Q_{\N[1]\times\N[2]}$ is
bounded by $4|Q_{1}||Q_{2}|$.
W.l.o.g., assume also that $\bot$ and $\top$ do not occur in $Q_1\cup Q_2$.%
\footnote{Or just pretend that $\bot$ and $\top$ are treated a normal states, this
is just to avoid trivial special cases such as $\Pair{\set{q_1}}{\emptyset}$.}
\begin{cor}[Product]
\label{cor:Product}
\label{cor:prod}
Let $\Pair{U}{V}\,{\in}\,Q_{\N[1]{\times}\N[2]}$. Then 
$U\,{\cap}\, V=\emptyset$,
$U\,{\cap}\, (F_1{\cup}F_2) = \emptyset$, and
there exist $q_1\,{\in}\, Q_1$ and $q_2\,{\in}\, Q_2$
such that 
$
U\,{\cup}\, V=\{q_1,q_2\}$
and
$
\Lang[{\N[1]{\times}\N[2]}]{\Pair{U}{V}} = \Lang[{\N[1]}]{q_1}\,{\cap}\, \Lang[{\N[2]}]{q_2}$.
\end{cor}
A further improvement is that the four cases of $\Pair{U}{V}$ that arise in
the corollary can be reduced to two cases by lifting of the classical product of
$\NBA$s \cite{Cho74} (see \cite[Theorem~2]{Kupf18}) to modulo $\A$,
which improves the upper bound to $2|Q_1||Q_2|$.

The construction of product, yet again, demonstrates that
mintermization is avoidable.
Moreover, a classical product $\mt{\N[1]}\times\mt{\N[2]}$ is in general \emph{not meaningful}
because the respective finite alphabets need not be the same.
One would first need to determine the \emph{shared}
alphabet $\Minterms{\CondOf{\N[1]}\cup\CondOf{\N[2]}}$,
which can have an up-front \emph{exponential} cost.
For example, if $\N[1]$ and $\N[2]$ each use 10 predicates from $\PA$
then the shared alphabet can
in the worst case be of size $2^{20}$.

\section{LTL Modulo $\A$}
\label{sec:ltl}

Here we lift classical linear temporal logic (LTL) to be modulo
$\A=(\D, \Psi, \den{\_}, \bot, \top, \orA, \andA, \notA)$ as a given
effective Boolean algebra, via symbolic derivatives that translate
LTL into transition terms. We then utilize the
algorithms developed for $\ABA[\A]$ for decision procedures of $\LTL[\A]$. 
Let $\alpha$ range over $\PA$ and $\varphi,\psi$ range over $\LTL$.
\[
\varphi \quad::=\quad
\alpha \quad\mid\quad
\lnot\varphi \quad\mid\quad
\varphi_1\lor\varphi_2 \quad\mid\quad
\varphi_1\land\varphi_2 \quad\mid\quad
\Next\varphi \quad\mid\quad
\varphi_1\Until\varphi_2 \quad\mid\quad
\varphi_1\Release\varphi_2
\]
where $\Next$ is \emph{Next}, $\Until$ is \emph{Until}, and $\Release$
is \emph{Release}, are called \emph{modal} formulas. We let the \emph{true} 
formula be $\top$ and
the \emph{false} formula be $\bot$ from $\PA$, and use the following
standard abbreviations:
\[
\varphi \limplies \psi \eqdef \lnot\varphi\lor\psi,
\quad
\Finally\psi\eqdef\top\Until\psi,
\quad
\Globally\psi\eqdef \bot\Release\psi
\]
where $\Finally$ is called \emph{Eventually (or Finally)} and
$\Globally$ is called \emph{Globally (or Always)}.

\subsection{Semantics}

A word $w\in\Do$ is a \emph{model of} $\varphi\in\LTL[\A]$,
$w\models\varphi$, when the following holds, where $\alpha\in\PA$:
\begin{eqnarray}
\label{eq:M1}
w\models\alpha &\eqdef& \ith{w}{0}\in\den{\alpha}\\
\label{eq:M2}
w\models\varphi\land\psi &\eqdef& w\models\varphi \band w\models\psi \\
\label{eq:M3}
w\models\varphi\lor\psi &\eqdef& w\models\varphi \bor w\models\psi \\
\label{eq:M4}
w\models\lnot\varphi &\eqdef& w\not\models\varphi \\
\label{eq:M5}
w\models \Next\psi &\eqdef& \Rest{w}\models\psi \\
\label{eq:M6}
w\models \varphi\Until\psi &\eqdef& \exists j:
\Rest[j]{w}\models\psi\;\textrm{and}\;\forall i{<}j:\Rest[i]{w}\models\varphi \\
\label{eq:M7}
w\models\varphi\Release\psi &\eqdef&
\forall j: \Rest[j]{w}\models\psi \bor 
\exists j:\Rest[j]{w}\models\varphi\band\forall i{\leq}j:\Rest[i]{w}\models\psi \\
\label{eq:L}
\Lang{\varphi} &\eqdef& \{w\in\Do\mid w\models\varphi\}
\end{eqnarray}
Let $\varphi\equiv\psi\eqdef\Lang{\varphi}=\Lang{\psi}$.
The rules (\ref{eq:M6}) and (\ref{eq:M7}) are duals of each other, either one suffices as the
main definition.
Observe that if $\alpha\in\PA$
then $w \models \lnot\alpha$ iff $w \not\models \alpha$ iff
$\ith{w}{0} \notin \den{\alpha}$ iff
$\ith{w}{0} \in \D\setminus\den{\alpha}$ iff 
$w \models \notA\alpha$.
If follows that any Boolean combination of predicates from $\PA$ can itself
be converted into a predicate in $\PA$.

The following examples illustrate some cases of $\LTL[\A]$ modulo
various $\A$.  The first example illustrates --
\emph{at a very abstract level} -- the well-known connection of
integrating SAT solving into \emph{symbolic LTL}, e.g.,
by using BDDs \cite{WulfDMR08}.

\begin{ex}
\label{ex:classicalLTL}
Classical LTL over a set of atomic propositions $P$ is 
$\LTL[\A]$ where $\A$ is an algebra of Boolean combinations over $P$, where $\D=2^P$
and $\Psi=\BC{P}$, and 
where satisfiability can be implemented using a SAT solver.
An element $d\in\D$ such that $d\vDash\alpha$
defines a \emph{truth
assignment} to $P$  that makes $\alpha$ true.  For example, if
$P=\{p_i\}_{i<7}$ and $\alpha = p_6\andA p_5\andA p_4\andA (p_3 \orA
((\notA p_2)\andA p_1))$ then if $w\in\Do$ is such
that $w(0)=\{p_1,p_4,p_5,p_6\}$ and $w(1) = \{p_1,p_2,p_4,p_5,p_6\}$
then $w(0)\models\alpha$ but $w(1)\not\models\alpha$.
Thus, e.g., $w\not\models \Globally\alpha$.
\end{ex}

While in classical LTL, as in Example~\ref{ex:classicalLTL},
$\D$ is \emph{finite},
in the next example $\D$ is \emph{infinite}.

\begin{ex}
  \label{ex:smt}
  Consider LTL modulo $\A$ as
  the algebra of Satisfiability Modulo Theories formulas $\PA$
  that can be implemented using an SMT solver.
  In this case $\D$ is the set of models that provide interpretations to \emph{uninterpreted constants} or \emph{variables}.
  Let $\beta$ be the predicate $0 < x$ and let $\alpha$ be
  the predicate $x < 1$ where $x$ is a variable of type \emph{real}.
  Then $\alpha\Release\beta$ states that $x$
  must remain positive until $x<1$.
  E.g., $(x{\mapsto}1)(x{\mapsto}\frac{1}{2})(x{\mapsto}0)^\omega\models\alpha\Release\beta$.
\end{ex}

\subsection{Vardi Derivatives of $\LTL$ Lifted to Modulo $\A$}
\label{sec:deriv}

Here we show how the semantics of $\LTL[\A]$ can be realized via transition terms.
The key observation here is that
a concrete derivative (i.e., for a given symbol $a\in\D$)
is not actually constructed but
maintained in a symbolic form as a transition term. 
Let $\alpha$ range over $\PA$, and $\varphi$ and $\psi$ range over $\LTL[\A]$.
The \emph{symbolic derivative} of an $\LTL[\A]$ formula is defined as
the following transition term in $\TT[\A]{{\LTL}}$:
\begin{eqnarray}
\label{eq:d1}
  \deriv{\alpha} &\eqdef& \ifthen{\alpha}{\top}\\
  \label{eq:d2}
  \deriv{\varphi \land \psi} &\eqdef& \deriv{\varphi} \land \deriv{\psi}\\
  \label{eq:d3}
  \deriv{\varphi \lor \psi} &\eqdef& \deriv{\varphi} \lor \deriv{\psi}\\
  \label{eq:d4}
  \deriv{\lnot\varphi} &\eqdef& \lnot\deriv{\varphi}\\
  \label{eq:d5}
  \deriv{\Next\psi} &\eqdef& \psi\\
  \label{eq:d6}
  \deriv{\varphi\Until\psi} &\eqdef& \deriv{\psi}\lor(\deriv{\varphi}\land(\varphi\Until\psi)) \\
  \label{eq:d7}
  \deriv{\varphi\Release\psi} &\eqdef& \deriv{\psi}\land(\deriv{\varphi}\lor(\varphi\Release\psi))
\end{eqnarray}
Symbolic derivatives, as given above, lift the construction in \cite{Vardi95LTL}
(see also \cite[Theorem 24]{Kupf18}) as to be modulo $\A$.

We also let $\deriv{\top}\eqdef\top$ and $\deriv{\bot}\eqdef\bot$ and
observe that $\lnot\top\equiv\bot$ and $\lnot\bot\equiv\top$.
We use the following properties of
$\Globally$ and $\Finally$ (where $\bot\lor\phi\feq\phi$ and
$\top\land\phi\feq\phi$ are applied implicitly):
\[
  \deriv{\Globally\psi} \feq \deriv{\psi}\land\Globally\psi
  \qquad
  \deriv{\Finally\psi} \feq \deriv{\psi}\lor\Finally\psi
\]
Correctness of the derivation rules will follow as a special case
of Theorem~\ref{thm:ELTL} (presented in Section~\ref{sec:ltlere}).

We now link derivatives formally with $\ABA$s.
We write $\LTLp[\A]$ for the \emph{positive} formulas in
$\LTL[\A]$ where $\lnot$ does not occur.
This is a standard normal form assumption and every formula
in $\LTL$ can be translated into
$\LTLp$ of the same size if members of $\PA$
are treated as units. The particular aspect with modulo $\A$ is that
\emph{complement is propagated into $\A$}, i.e., for $\alpha\in\PA$,
the formula $\lnot\alpha$ becomes the predicate $\notA\alpha$.
\begin{definition}[{$\LTLp[\A]$}]
\label{def:LTLp}
 The \emph{positive fragment of $\LTL[\A]$} is obtained by
eliminating any explicit use of $\lnot$ through de Morgan's laws 
and the standard equivalence preserving
rules where $\alpha\in\PA$ and $\varphi,\psi\in\RLTL$:
$\lnot\alpha\equiv\notA\alpha$,
$\lnot\Next\varphi\equiv\Next\lnot\varphi$,
$\lnot(\varphi\Until\psi)\equiv\lnot\varphi\Release\lnot\psi$,
and
$\lnot(\varphi\Release\psi)\equiv\lnot\varphi\Until\lnot\psi$.
\end{definition}

\begin{definition}[{$\M[\phi]$ for $\phi\in\LTLp[\A]$}]
For $\phi\in\LTLp[\A]$ let $Q_\phi$ denote the set
containing $\bot$ and $\top$ and
all \emph{non-Boolean} subformulas of $\phi$, i.e., predicates in $\PA$,
and all \emph{modal} formulas.
Let $\tf_\phi$ denote $\lambda q.\deriv{q}$ for $q\in Q_\phi$.
Let $F_\phi$ contain all the $\Release$-formulas and $\top$ in $Q_\phi$.
Then
$
\M[\phi]\eqdef(\A,Q_\phi,\phi,\tf_\phi, F_\phi).
$
Observe that $\tf_{\phi}:Q_\phi\rightarrow\TT[\A]{\BCp{Q_\phi}}$.
\end{definition}

We now state a key result from \cite{Vardi95LTL}
(\cite[Theorem~24]{Kupf18}) lifted as to be modulo $\A$. 

\begin{thm}[Vardi derivatives modulo theories]
\label{thm:FA}
$\forall\phi\in\LTLp[\A]: \Lang{\M[\phi]}=\Lang{\phi}$
\end{thm}
\noindent
Theorem~\ref{thm:FA} implies that the more general language invariant
is preserved by \emph{all} states of $\M[\phi]$.
\begin{cor}[LTL invariance]
\label{cor:FA}
$\forall\phi\in\LTLp[\A]:\forall\psi\in Q_{\M[\phi]}: \Lang[{\M[\phi]}]{\psi}=\Lang{\psi}$
\end{cor}
\begin{proof}
$\Lang[{\M[\phi]}]{\psi}{=}\Lang{{\M[\phi]^{(\psi)}}}\stackrel{(\star)}{=}\Lang{{\M[\psi]}}{=}\Lang{\psi}$
where $(\star)$ holds because only $Q_{\psi}$ matters.
\end{proof}

\subsection{Examples}

\begin{ex}
\label{ex:GFaF!a}
Let $\phi = \Globally(\Finally\alpha\land\Finally\notA\alpha)$, where
$\alpha\in\PA$ is such that $\alpha\nequiv\bot$ and $\alpha\nequiv\top$.
Then the transition function $\tf=\tf_{\phi}$ has the following (relevant) transition terms
in addition to $\tf(\top)=\top$:
\[
\tf(\Finally\alpha) = \ite{\alpha}{\top}{\Finally\alpha},\quad
\tf(\Finally\notA\alpha) = \ite{\notA\alpha}{\top}{\Finally\notA\alpha},\quad
\tf(\phi) = \ite{\alpha}{\Finally\notA\alpha\land\phi}{\Finally\alpha\land\phi}
\]
where
$
\deriv{\Finally\alpha} \feq (\deriv{\alpha}\lor\Finally\alpha) =
(\ifthen{\alpha}{\top}\lor\Finally\alpha) = \ite{\alpha}{\top\lor\Finally\alpha}{\bot\lor\Finally\alpha}
\feq \ite{\alpha}{\top}{\Finally\alpha}
$ and similarly that $\deriv{\Finally\notA\alpha}\feq\ite{\notA\alpha}{\top}{\Finally\notA\alpha}$, and
\[
\begin{array}{@{}r@{\;}c@{\;}l@{}}
\deriv{\phi} &\feq& (\deriv{\Finally\alpha\land\Finally\notA\alpha}\land\phi)
= (\deriv{\Finally\alpha}\land\deriv{\Finally\notA\alpha}\land\phi)
\feq (\ite{\alpha}{\top}{\Finally\alpha}\land \ite{\notA\alpha}{\top}{\Finally\notA\alpha}\land\phi)
\\
&\feq& \ite{\alpha}{\ite{\notA\alpha}{\phi}{\Finally\notA\alpha\land\phi}}{
                \ite{\notA\alpha}{\Finally\alpha\land\phi}{\Finally\alpha\land\Finally\notA\alpha\land\phi}}
\feq
\ite{\alpha}{\Finally\notA\alpha\land\phi}{\Finally\alpha\land\phi}
\end{array}
\]
$\M[\phi]$ is shown in Figure~\ref{fig:aba}
with $q_0=\phi$, $q_1=\Finally\alpha$ and $q_2=\Finally\notA\alpha$.
\end{ex}

Note that some states may become \emph{irrelevant} (unreachable from $\phi$) through
rewrites of $\tf_{\phi}$, which the abstract definition of $\M[\phi]$ does not 
directly reflect, e.g., $\alpha$ and $\notA\alpha$ are not relevant as states in
Example~\ref{ex:GFaF!a} above.
The following example illustrates the impact
of cleaning modulo different element theories $\A$ and
how this is reflected in the computation of $\M[\phi]$.
\begin{ex}
\label{ex:smt2}
Consider $\A$ as an SMT solver, let $x$ be of type \emph{integer}, and
let $\phi=(x{<}1)\Release(0{<}x)$, stating that the condition that $x$ must be positive
is released when $x<1$ becomes true. Then
\[
\begin{array}{@{}r@{\;}c@{\;}l@{}}
\deriv{\phi} &\feq& \deriv{0{<}x}\land(\deriv{x{<}1}\lor\phi) =
\ite{0{<}x}{\top}{\bot}\land(\ite{x{<}1}{\top}{\bot}\lor\phi) \\
&\feq&
\ite{0{<}x}{\top}{\bot}\land\ite{x{<}1}{\top}{\phi}
\feq
\ite{0{<}x}{\ite{x{<}1}{\top}{\phi}}{\bot}
\feq
\ite{0{<}x}{\phi}{\bot}
\end{array}
\]
where $\UNSAT{(0{<}x){\andA}(x{<}1)}$ is used for cleaning.
Essentially $(x{<}1)\Release(0{<}x)$ becomes
$\Globally(0{<}x)$.
\end{ex}

\subsection{Alternation Elimination for $\LTL[\A]$}

Combined together, Theorem~\ref{thm:AE} and
Corollary~\ref{cor:FA} show
the decidability of $\LTL[\A]$ for effective $\A$
and allow us to directly use the alternation
elimination algorithm (modulo $\A$) on the transition terms 
resulting from the symbolic derivatives of $\LTL[\A]$
to produce an nondeterministic B\"uchi automata from
an LTL formula in a lazy manner, while safely applying
many LTL-based rewrites.

\begin{thm}[LTL invariance of {\AElimName}]
$\forall\phi\in\LTLp[\A]:\forall q\in Q_{\AElim{\M[\phi]}}:\Lang[\AElim{\M[\phi]}]{q}=\Lang{\toconj{q}}$
\end{thm}
\begin{proof}
Let $\M=\M[\phi]$ and $\N=\AElim{\M}$.
Fix $q=\Pair{U}{V}\in Q_{\N}$. We show that $\Lang[\N]{q}=\Lang{\toconj{q}}$.
We have that
$\Lang[\N]{q}=\Lang[\M]{\toconj{q}}= \bigcap_{\psi\in U\cup V}\Lang[\M]{\psi} =
\bigcap_{\psi\in U\cup V}\Lang{\psi} = \Lang{\toconj{q}}$
where the first equality holds by Theorem~\ref{thm:AE} and
the third equality holds by Corollary~\ref{cor:FA}
because $U\cup V\subseteq Q_{\M}$.
\end{proof}

Many optimizations can be applied to states $\Pair{U}{V}$ in
$Q_{\AElim{\M[\phi]}}$ above in accordance with
Lemma~\ref{lma:state-reduction}.  For example, for any $\varphi$, if
$\varphi,\Globally(\varphi\land\_)\in V$ then $\varphi$ can be
eliminated from $V$: observe that trivially
$\deriv{\varphi\land\Globally(\varphi\land\_)}\feq\deriv{\Globally(\varphi\land\_)}$,
which means that the derivative (transition term) does not even need
to be computed in this case, in order to know that
Lemma~\ref{lma:state-reduction} can be applied.  In fact, this exact
situation appeared twice implicitly in
Example~\ref{ex:alt-elim-sample}.  Many similar rules can be used,
where the one illustrated here, is an instance of the subsumption rule
in \cite[$\leq$]{cavSomenziB00}.  Such rules would clearly be out of
reach if the LTL invariance property had been lost in translation.
A full analysis of which rules in \cite{cavSomenziB00} preserve
LTL invariance of {\AElimName} is beyond the scope of this work.

\section{LTL with Extended Regular Expressions Modulo $\A$}
\label{sec:ltlere}

We now show the power of symbol derivatives   
by using them to combine the following two languages into one:
\begin{itemize}
\item $\ERE[\A]$, extended regular expressions modulo $\A$, as given by \cite{StanfordVB21}, summarized
in Section~\ref{sec:revERE};
\item $\LTL[\A]$, as given in the previous section.
\end{itemize}
We first work with \emph{finite} words in $\Ds$, and later
lift the semantics to \emph{infinite} words in $\Do$ when 
we extend $\LTL[\A]$ with $\ERE[\A]$ in the combined language
$\RLTL[\A]$ (Section~\ref{sec:combLTLERE})

Our focus here is on the \emph{classical subset} of the extended regular expression 
operators supported in SPOT \cite{spotRef} and PSL \cite{PSL}.\footnote{Many other
operators, such as \emph{fusion}, \emph{non-length-matching
intersection} and \emph{bounded iteration} can also be supported.} 
The language $\RLTL[\A]$ has regex \emph{complement}, which is not in SPOT/PSL
but can be supported naturally with transition terms.

Finally, we lift $\omega$-regular languages to be modulo $\A$ and show that 
$\RLTL[\A]$ precisely captures $\omega$-regularity modulo $\A$ (Section~\ref{sec:omegaRLTL}).

\subsection{Extended Regular Expressions Modulo $\A$ and Their Derivatives}
\label{sec:revERE}

We fix $\A=(\D, \Psi, \den{\_}, \bot, \top, \orA, \andA, \notA)$.  
$\ERE[\A]$ or $\ERE$ for short is defined by the following abstract grammar, where $\alpha$ ranges
over $\PA$. We let $R$ range over $\ERE$, $R$ is called a \emph{regex}.
We write $\RE[\A]$ for the \emph{standard} fragment of $\ERE[\A]$ without $\inter$ or $\compl$.
\[
R \quad::=\quad \alpha \quad|\quad \eps
\quad|\quad  R_1 \union R_2
\quad|\quad R_1 \inter R_2
\quad|\quad R_1 \conc R_2
\quad|\quad R\st
\quad|\quad \compl R
\]
Here $\eps$ stands for the regex that accepts \emph{the empty word} $\epsilon\in\Ds$.
We reuse $\bot$ from $\PA$ for the regex that that
accepts \emph{nothing} (\emph{the empty language}).  The operators
appear in order of precedence, with
\emph{union} $(\union)$ having lowest and \emph{complement} $(\compl)$ having highest precedence.
The remaining operators are \emph{intersection} $(\inter)$,
\emph{concatenation} $(\conc)$,
and \emph{Kleene star} $(\st)$.
We also let $R\plus\eqdef R{\conc}R\st$.
We often abbreviate concatenation by juxtaposition.
Then $(\Ds,\ERE,\FLangName,\bot,\top\st,\union,\inter,\compl)$
is the associated effective Boolean algebra,
where, for $\alpha\in\PA$ and $R,R'\in\ERE$, the following additional laws hold:
\[
\begin{array}{l}
\FLang{\alpha}\eqdef\den{\alpha}\qquad \FLang{\eps}\eqdef\{\epsilon\}\qquad
\FLang{R\conc R'} \eqdef \FLang{R}{\cdot}\FLang{R'}\qquad
\FLang{R\st}\eqdef\FLang{R}^*
\end{array}
\]
Observe therefore that $\FLang{\compl\bot}=\FLang{\top\st}=\Ds$ while
$\FLang{\notA\bot}=\FLang{\top}=\D$.
The following operation, called \emph{nullability} of a regex $R$, is
defined recursively to decide if $\epsilon\in \FLang{R}$:\footnote{Nullability
can be maintained as a Boolean flag associated with the abstract syntax tree
node of every regex that is maintained at construction time of regexes and
therefore nullability test is a trivial constant time operation at runtime.}
\[
\begin{array}{@{}r@{\;}c@{\;}lr@{\;}c@{\;}lr@{\;}c@{\;}l@{}}
\Null{\eps}&\eqdef&\True&
\Null{R\st}&\eqdef&\True &
\Null{R \union R'}&\eqdef&\Null{R}\bor \Null{R'}\\
&&& \Null{\alpha}&\eqdef&\False & \Null{R \inter R'}&\eqdef&\Null{R}\band \Null{R'}\\
&&&\Null{\compl R}&\eqdef&\bnot\,\Null{R} &  \Null{R \conc R'}&\eqdef&\Null{R}\band \Null{R'}\\
\end{array}
\]
We now recall
the definition of \emph{transition regexes} from \cite{StanfordVB21}\footnote{Here
we use $\derName$ not to overload $\derivName$ and to avoid ambiguity because
$\der{\alpha}\neq\deriv{\alpha}$ for $\alpha\in\PA$.}
\begin{equation}
\label{eq:der}
\begin{array}{@{}r@{\;}c@{\;}l@{\;}r@{\;}c@{\;}ll@{}}
\der{\eps} &\eqdef& \bot & \der{\compl R} &\eqdef& \compl\der{R} &
\multirow{3}{*}{\mbox{$
\der{R{\conc}R'} \eqdef \ITEBrace{\Null{R}}{\der{R}{\conc}R'\,{\union}\,\der{R'}}{\der{R}{\conc}R'}
$}}\\
\der{\alpha} &\eqdef& \ifthen{\alpha}{\eps} & \der{R\,{\union}\,R'} &\eqdef& \der{R}\,{\union}\,\der{R'}\\
\der{R\st}& \eqdef& \der{R}{\conc}R\st\quad & \der{R\,{\inter}\,R'} &\eqdef& \der{R}\,{\inter}\,\der{R'}
\end{array}
\end{equation}
where $\union$ and $\inter$ are lifted to transition terms in $\TT{\ERE}$ by
(\ref{eq:binary}), and $\lambda x.x{\conc}r$ and $\compl$ are lifted
by (\ref{eq:unary}).  We use Lemma~\ref{lma:der} that follows
from \cite[Theorem~4.3 and Theorem~7.1]{StanfordVB21}.  In this
context functional equivalence $f\feq g$ based rewrite rules of
transition regexes is also crucial for termination of the resulting
DFA $\M[R]$ defined below.  In addition to \emph{associativity,
commutativity and idempotence} of $\union$ and $\inter$,
de Morgan's laws are used to propagate $\compl$, and concatenation is
an \emph{associative operator with $\eps$ as unit and $\bot$ as zero}.
Most of these rules are \emph{mandatory} for $\M[R]$ to
be \emph{finite} and several others are practically
very useful to further reduce the state space. E.g.,
$\compl\bot\feq \top\st$ and $\compl\eps\feq\top\plus$.
\begin{definition}[{$\M[R]$ for $R\in\ERE[\A]$}]
$\M[R]=(\A,Q,R,\tf,F)$ where
$Q\subset\ERE$ is finite, $R\in Q$ is the initial state, 
$\tf:Q\rightarrow\TT{Q}$ is the transition function such that
$\tf(q)\feq\der{q}$, and $F=\{q\in Q\mid\Null{q}\}$.
\end{definition}
Let $\M=(\A,Q,\Qi,\tf,F)$.
For $u\in\Ds$ and $q\in Q$, let $\tfs{u}{q}$ be the state reached from $q$ by $u$.
\[
\tfs{u}{q} \eqdef \IfThenElse{u=\epsilon}{q}{\tfs{\Rest{u}}{\leafof{\tf(q)}{\ith{u}{0}}}}
\]
The \emph{language of $\M$ for $q\in Q$} is then
$\FLang[\M]{q}\eqdef\{u\in\Ds\mid \tfs{u}{q}\in F\}$.
\begin{lma}
\label{lma:der}
For $\M=\M[R]$ as above,
$\forall q\in Q,a\in\D:
\FLang[{\M}]{q}=\FLang{q} \;\textrm{and}\;
\FLang{\leafof{\tf(q)}{a}}=\DERIV[a]{\FLang{q}}$.
\end{lma}

A state $r$ is called \emph{alive} if $\FLang[\M]{r}\neq\emptyset$ else \emph{dead}.
Observe that for $\M[R]$, a state
$\tfs{u}{R}$ is alive iff $\exists
x:ux\in\FLang{R}$, i.e., the regex reached from $R$ via $u$ accepts
some $x\in\Ds$.

Let $\One{R}$ denote the condition in $\PA$ that corresponds to
$\FLang{R}\cap\D$. Formally, let $R,r\in\ERE$, 
\[
\begin{array}{rcl}
\One{R} &\eqdef& \One[\top]{\der{R}} \\
\One[\beta]{r} &\eqdef& \IfThenElse{\Null{r}}{\beta}{\bot} \\
\One[\beta]{\ite{\alpha}{f}{g}} &\eqdef& \One[\beta\andA\alpha]{f}\orA\One[\beta\andA\notA\alpha]{g}
\end{array}
\]
Moreover, normalize $\One{R}$ to $\bot$ if $\One{R}\equiv\bot$ and to $\top$ if
$\One{R}\equiv\top$.

\subsection{LTL Combined with Extended Regular Expressions Modulo $\A$}
\label{sec:combLTLERE}

Here we consider the following extension $\RLTL[\A]$ of $\LTL[\A]$ that combines
$\ERE$ with $\LTL$. Let $\phi$ range over $\RLTL$ and let $R$ range over $\ERE$.
Let $\alpha$ range over $\PA$.
\begin{eqnarray*}
\phi &::= &
\alpha \;\mid\;
\lnot\phi \;\mid\;
\phi\lor\phi' \;\mid\;
\phi\land\phi' \;\mid\;
\Next\phi \;\mid\;
\phi\Until\phi' \;\mid\;
\phi\Release\phi' \;\mid\;
\\
&& 
R\eimpl\phi \;\mid\;
R\uimpl\phi \;\mid\;
\wcl{R}\;\mid\;
\nwcl{R}\;\mid\;
\ocl{R}
\end{eqnarray*}
The operator
$\eimpl$ is \emph{existential suffix implication},
$\uimpl$ is \emph{universal suffix implication},
$\wcl{R}$ is the \emph{weak closure} of $R$,
$\nwcl{R}$ is the \emph{negated weak closure} of $R$,
and
$\ocl{R}$ is the \emph{$\omega$-closure} of $R$.
The semantics of $\phi\in\RLTL$ is
a conservative extension of $\LTL$ and the
semantics of the new formulas is consistent 
with the semantics in \cite[2.6.1]{spotRef}, except that if $R$ is \emph{nullable}
then in $\RLTL$ the formula $\wcl{R}$ (resp.~$\nwcl{R}$) is equivalent to $\top$
(resp.~$\bot$) and $\ocl{R}$ is not supported in SPOT.\footnote{In other words,
$\wcl{R}$ (resp.~$\nwcl{R}$) in SPOT corresponds in $\RLTL$ 
to $\wcl{R\inter\top\plus}$ (resp.~$\nwcl{R\inter\top\plus}$).}
Let $w\in\Do$.
\begin{eqnarray}
\label{eq:eimpl}
w\models R\eimpl\phi &\eqdef& \exists\, u,v:  w=uv \band u\ith{v}{0}\in\FLang{R} \band v\models\phi \\
\label{eq:uimpl}
w\models R\uimpl\phi &\eqdef& \forall\, u,v: \IfThen{w=uv\band u\ith{v}{0}\in\FLang{R}}{v\models\phi} \\
\label{eq:wcl}
w\models \wcl{R} &\eqdef&
(\exists\,u\prece w: u \in\FLang{R}) \bor
(\forall\,u\prec w:\exists\,x:ux\in\FLang{R}) \\
\label{eq:nwcl}
w\models \nwcl{R} &\eqdef& w\not\models \wcl{R} \\
w\models \ocl{R} &\eqdef& w \in \ocl{(\FLang{R}\setminus\{\epsilon\})}
\end{eqnarray}
Definitions (\ref{eq:M1})-(\ref{eq:L}) carry over to $\RLTL$.
For $\phi,\psi\in\RLTL$ let $\phi\equiv\psi \eqdef \Lang{\phi}=\Lang{\psi}$.
Observe that if $R$ is nullable in (\ref{eq:wcl}) then
$u=\epsilon\in\FLang{R}$ and thus $\wcl{R}\equiv\top$ and
$\nwcl{R}\equiv\bot$ by (\ref{eq:nwcl}).
It follows from the definitions that $R\uimpl\phi \equiv \lnot(R\eimpl\lnot\phi)$
and trivially that $\nwcl{R}\equiv\lnot\wcl{R}$. These formulas play a useful
role in the positive fragment of $\RLTL$ that we denote by $\RLTLp$.
\begin{definition}[$\RLTLp$] The \emph{positive fragment of $\RLTL$} is obtained by
using the rules in Definition~\ref{def:LTLp}, and for $R\in\ERE$ and $\varphi,\psi\in\RLTL$:
$\lnot(R\eimpl\varphi)\equiv R\uimpl\lnot\varphi$,
$\lnot(R\uimpl\varphi)\equiv R\eimpl\lnot\varphi$,
$\lnot\wcl{R}\equiv\nwcl{R}$, and
$\lnot\nwcl{R}\equiv\wcl{R}$.
Moreover, \emph{$\lnot\ocl{R}$ is not permitted in $\RLTLp$}.
\end{definition}
We now define transition terms for $\RLTL$ as members of $\TT{\RLTL}$.
Derivative rules
(\ref{eq:d1})--(\ref{eq:d7}) carry over to $\RLTL$ and we add the
following new rules:
\begin{eqnarray}
\label{eq:eimpl-deriv}
\deriv{R\eimpl\phi} &\eqdef& \deriv{\One{R}\land\phi} \lor (\der{R}\eimpl\phi) \\
\label{eq:uimpl-deriv}
\deriv{R\uimpl\phi} &\eqdef& \deriv{\One{R}\limplies\phi}\land(\der{R}\uimpl\phi) \\
\label{eq:wcl-deriv}
\deriv{\wcl{R}} &\eqdef& \IfThenElse{\Null{R}}{\top}{\wcl{\der{R}}}\\
\label{eq:nwcl-deriv}
\deriv{\nwcl{R}} &\eqdef& \IfThenElse{\Null{R}}{\bot}{\nwcl{\der{R}}}\\
\label{eq:ocl-deriv}
\deriv{\ocl{R}} &\eqdef& \deriv{R \eimpl \Next\ocl{R}}
\end{eqnarray}
where 
$r\,{\eimpl}\,\phi$,
$r\,{\uimpl}\,\phi$,
$\wcl{r}$, and
$\nwcl{r}$
are lifted to
$\TT{\RLTL}$ by (\ref{eq:unary}).
For $u\in\Ds$ and $\phi\in\RLTL$ let
\[
\derivs{u}{\phi}\eqdef \IfThenElse{u=\epsilon}{\phi}{\derivs{\Rest{u}}{\leafof{\deriv{\phi}}{\ith{u}{0}}}}.
\]
Rules (\ref{eq:eimpl-deriv})--(\ref{eq:nwcl-deriv}) seamlessly combine
the derivatives of $\LTL$ (\ref{eq:d1})--(\ref{eq:d7}) with
derivatives of $\ERE$ (\ref{eq:der}).  Correctness of the combined
rule-set is proved in Theorem~\ref{thm:ELTL}.  Note that if
$\One{R}=\bot$ then the lhs disjunct of (\ref{eq:eimpl-deriv}) reduces
to $\bot$ and the lhs conjunct of (\ref{eq:uimpl-deriv}) reduces to
$\top$.  Moreover, if $\der{R}=\bot$ or $\der{R}=\eps$ then the rhs
disjunct of (\ref{eq:eimpl-deriv}) reduces $\bot$ and the rhs conjunct
of (\ref{eq:uimpl-deriv}) reduces to $\top$, using \emph{new} rewrite
rules in $\RLTL$: for $r=\bot$ or $r=\eps$, we let
$r\eimpl\phi\feq\bot$ and $r\uimpl\phi\feq\top$. Furthermore,
$\wcl{\bot}\feq\bot$ and
$\nwcl{\bot}\feq\top$,
and if $r$ is nullable then $\wcl{r}\feq\top$
and $\nwcl{r}\feq\bot$.
Observe that all these rules are consistent with
the semantics of $\RLTL$. 

Let $\Acc{\RLTLp}$ denote the subset of all the following formulas
of $\RLTLp$: $\top$,  all $\Release$ formulas,
all $\uimpl$ formulas,
all $\wcl{R}$ formulas where $R$ is \emph{alive},
all $\nwcl{R}$ formulas where $R$ is \emph{dead},
and all $\ocl{R}$ formulas.
\begin{definition}[{{$\M[\phi]$ for $\phi\in\RLTLp[\A]$}}]
\label{def:M-RLTLp}
$\M[\phi]$ is the $\ABA$ $(\A,Q,\phi,\tf,Q\cap\Acc{\RLTLp})$
where $Q\subseteq\RLTLp$ is the least set of states
such that $\phi\in Q$ and if
$q\in Q$ then $\tf(q)\feq\deriv{q}$ and
if $p$ is a non-Boolean subformula of a leaf of $\tf(q)$ then $p\in Q$.
\end{definition}
$\M[\phi]$ is well-defined ($Q$ above is finite)
because all the $\M[R]$ are finite where $R\in\ERE[\A]$ by Lemma~\ref{lma:der}.
In contrast to $\LTL$ derivatives, where all the original subformulas are enough,
for $\phi\in\RLTL$, in $\M[\phi]$ we also get new formulas in suffix implications and closures, e.g.,
in every leaf $r\eimpl\psi$ of $\der{R}\eimpl\psi$, $r$ is some derivative of $R$ and
$r\eimpl\psi$ is typically not a subformula of the original formula $\phi$.

\begin{ex}
Consider the $\RLTL$ formula $\psi = \alpha\eimpl\phi$
with $\phi=\Globally(\Finally\alpha \land \Finally\notA\alpha)$.
Then $\der{\alpha}=\ifthen{\alpha}{\eps}$. Since $\One{\alpha}=\alpha$, by using (\ref{eq:eimpl-deriv}),
we get that, recall also the derivatives from Example~\ref{ex:GFaF!a},
\[
\begin{array}{l|l|l}
q & \deriv{q} & \ABA\;\M[\psi] \\ \hline
\psi & \deriv{\alpha\land\phi}\lor\ifthen{\alpha}{\eps\eimpl\phi}&
\multirow{5}{*}{
\mbox{
\mytikz{-1.4em}{-1em}{-4.8cm}{-0.2cm}{
    \begin{tikzpicture}[buchi]
      \small
      \begin{scope}[on above layer]
        \coordinate (init) at (0,0);
        \coordinate (andX) at (2.5cm,0);
        \coordinate (center) at (3.3cm,0);
          \node[boolop,right=1cm of init] (and0) {\textrm{\scriptsize$\land$}};
          \node[boolop,above=0.4cm of andX] (and2) {\textrm{\scriptsize$\land$}};
          \node[boolop,below=0.4cm of andX] (and1) {\textrm{\scriptsize$\land$}};
          \node[boolop,above=0.4cm of andX] (and2) {\textrm{\scriptsize$\land$}};
          \node[boolop,above=0.4cm of andX] (and2) {\textrm{\scriptsize$\land$}};
          \node[state,right=0.2cm of init] (psi) {$\psi$};
          \node[state,accept,right=1.5cm of init] (q0) {$\phi$};
          \node[state,below=0.5cm of center] (q1) {$\Finally\alpha$};
          \node[state,above=0.5cm of center] (q2) {$\Finally\notA\alpha$};
          \node[state,accept,right=1cm of center] (top) {$\top$};   
      \end{scope}
      \draw[move] (init) -- (psi);
      \draw[move] (psi) -- node[above] {$\alpha$} (and0);
      \draw[move] (and0) -- (q0);
      \draw[macc] (and0) to[out=90,in=160] (q2);
      \draw[move] (q0) -- node[above] {$\notA\alpha$} (and1);
      \draw[move] (and1) -- (q1);
      \draw[macc] (and1) to[out=-90,in=-90] (q0);
      \draw[move] (q1) edge[loop right] node[right] {$\notA\alpha$} (q1);
      \draw[move] (q0) -- node[above] {$\alpha$} (and2);
      \draw[move] (and2) -- (q2);
      \draw[macc] (and2) to[out=90,in=90] (q0);
      \draw[move] (q2) edge[loop right] node[above] {$\alpha$} (q2);
      \draw[macc] (q2) -- node[left] {$\notA\alpha$} (top);
      \draw[macc] (q1) -- node[left] {$\alpha$} (top);
      \draw[move] (top) edge[loop above] node[above] {$\top$} (top);
    \end{tikzpicture}
}
}
}
\\
&
\feq \deriv{\alpha\land\phi}\lor\ifthen{\alpha}{\bot}
&
\\
&
\feq \ifthen{\alpha}{\top}\land\deriv{\phi}
&
\\
&
\feq
\ifthen{\alpha}{\ite{\alpha}{\Finally\notA\alpha\land\phi}{\Finally\alpha\land\phi}}
&
\\
&
\feq
\ifthen{\alpha}{\Finally\notA\alpha\land\phi}
&
\end{array}
\]
The resulting $\ABA$ requires alternation elimination for $\NBA$ conversion.
\end{ex}

\begin{ex}
Let $\alpha,\beta,\gamma\in\PA$ and consider
$\phi = (\alpha\beta)\plus\eimpl\Globally\gamma$. First we
construct the transition terms of $\M[(\alpha\beta)\plus]$
for all $r\in Q_{\M}$.
Note that $\eps$ (resp.~$\bot$) is the unit (resp.~zero) of concatenation.
For example, $\ifthen{\alpha}{\eps}{\conc}{\beta(\alpha\beta)\st}=
\ite{\alpha}{\eps\conc\beta(\alpha\beta)\st}{\bot\conc\beta(\alpha\beta)\st}\feq
\ite{\alpha}{\beta(\alpha\beta)\st}{\bot}=\ifthen{\alpha}{\beta(\alpha\beta)\st}$.
\[
\begin{array}{@{}l|l|l|l|l@{}}
r & \der{r} & \One{r} & \Null{r} & \DFA\;\M[(\alpha\beta)\plus] \\\hline
r_0=(\alpha\beta)\plus & \der{\alpha\beta(\alpha\beta)\st} =
\ifthen{\alpha}{\eps}{\conc}{\beta(\alpha\beta)\st} \feq
\ifthen{\alpha}{\beta(\alpha\beta)\st} &
\bot & \False &
\multirow{3}{*}{
\textrm{
\mytikz{-1em}{-.5em}{-1em}{-1em}{
    \begin{tikzpicture}[buchi]
      \small
      \begin{scope}[on above layer]
        \coordinate (init) at (0,0);
        \coordinate (center) at (1cm,0);
          \node[state,right=0.2cm of init] (r0) {$r_0$};
          \node[state,right=0cm of center] (r1) {$r_1$};
          \node[state,accept,below=0.2cm of r1] (r2) {$r_2$};
      \end{scope}
      \draw[move] (init) -- (r0);
      \draw[move] (r0) -- node[above] {$\alpha$} (r1);
      \draw[macc] (r2) to[out=180,in=200] node[left] {$\alpha$} (r1);
      \draw[macc] (r1) to[out=0,in=20] node[right] {$\beta$} (r2);
    \end{tikzpicture}
}
}
}
\\
r_1 = \beta(\alpha\beta)\st & \ifthen{\beta}{(\alpha\beta)\st} &
\beta & \False &
\\
r_2 = (\alpha\beta)\st & \der{\alpha\beta}{\conc}(\alpha\beta)\st =
\ifthen{\alpha}{\beta}{\conc}(\alpha\beta)\st \feq
\ifthen{\alpha}{\beta(\alpha\beta)\st} &
\bot & \True &
\end{array}
\]
The transition terms that arise from $\phi$
give rise here to the $\ABA$ $\M[\phi]$ as shown below.
Observe that $\M[\phi]$ also happens to be an $\NBA$ because
conjunctions do not arise:
\[
\begin{array}{@{}l|l|l@{}}
q & \deriv{q} & \ABA\;(\NBA)\;\M[\phi] \\\hline
q_0 =  r_0\eimpl\Globally\gamma & \deriv{\bot\land\Globally\gamma}\lor
(\ifthen{\alpha}{r_1}\eimpl\Globally\gamma)
\feq \ifthen{\alpha}{r_1\eimpl\Globally\gamma}
&
\multirow{6}{*}{
\textrm{
\mytikz{-1.5em}{-.5em}{-0.5em}{-1em}{
    \begin{tikzpicture}[buchi]
      \small
      \begin{scope}[on above layer]
        \coordinate (init) at (0,0);
        \coordinate (center) at (1cm,0);
          \node[state,right=0.2cm of init] (q0) {$q_0$};
          \node[state,right=0.2cm of center] (q1) {$q_1$};
          \node[state,below=0.7cm of q1] (q2) {$q_2$};
          \node[state,accept,right=1.5cm of center] (q3) {$q_3$};   
      \end{scope}
      \draw[move] (init) -- (q0);
      \draw[move] (q0) -- node[above] {$\alpha$} (q1);
      \draw[move] (q3) edge[loop above] node[right] {$\gamma$} (q3);
      \draw[macc] (q1) -- node[above] {$\beta{\andA}\gamma$} (q3);
      \draw[macc] (q2) to[out=110,in=250] node[left] {$\alpha$} (q1);
      \draw[macc] (q1) to[out=-70,in=70] node[right] {$\beta$} (q2);
    \end{tikzpicture}
}
}
}
\\
q_1 = r_1\eimpl\Globally\gamma &
\deriv{\beta\land\Globally\gamma}\lor
\ifthen{\beta}{r_2\eimpl\Globally\gamma} 
&  \\
&
\feq \ifthen{\beta{\andA}\gamma}{\Globally\gamma}
\lor
\ifthen{\beta}{r_2\eimpl\Globally\gamma} \\
& \feq \ite{\beta{\andA}\gamma}{\Globally\gamma\lor
(r_2\eimpl\Globally\gamma)}{
\ifthen{\beta}{r_2\eimpl\Globally\gamma}}
 \\
q_2 = r_2\eimpl\Globally\gamma & \ifthen{\alpha}{r_1\eimpl\Globally\gamma} & \\
q_3 = \Globally\gamma & \ifthen{\gamma}{\Globally\gamma} & 
\end{array}
\]
Note that the guard $\beta$ of the transition $\tr{q_1}{\beta}{q_2}$ 
is a simplification of $(\beta{\andA}\gamma){\orA}(\notA(\beta{\andA}\gamma){\andA}\beta)$.
The automaton is thus truly \emph{nondeterministic} because the outgoing transition guards from
$q_1$ overlap, in other words, the leaf
$\Globally\gamma\lor (r_2\eimpl\Globally\gamma) = q_3\lor q_2$ of $\deriv{q_1}$ is reachable
when $\beta{\andA}\gamma\nequiv\bot$.
\end{ex}

\begin{ex}
\label{ex:wcl}
Let $\alpha$ and $\beta$ be predicates such that $\den{\alpha}=\{a\}$
and $\den{\beta}=\{b\}$ where $a\neq b$.
Let $R={(\alpha{\conc}\top)\st{\conc}\beta}$.
The following table shows the
transitions that arise in the $\DBA$ $\M[\wcl{R}]$,
where $\UNSAT{\alpha{\andA}\beta}$ is used to clean $\deriv{\wcl{R}}$.
Observe that $\notA\alpha{\andA}\beta \equiv \beta$.
The table also shows $\M[\nwcl{R}]$ where
$\nwcl{\ifthen{\beta}{\eps}}=\ite{\beta}{\nwcl{\eps}}{\nwcl{\bot}}\feq
\ite{\beta}{\bot}{\top}\feq\ifthen{\notA\beta}{\top}$.
\[
\begin{array}{@{}l|l|l|l@{}}
q & \deriv{q} & \DBA\;\M[\wcl{R}] & \DBA\;\M[\nwcl{R}] \\\hline
\wcl{R} &
\ite{\alpha}{\wcl{\top{\conc}R}}{\ifthen{\beta}{\wcl{\eps}}}
&
\multirow{3}{*}{
\textrm{
\mytikz{-2em}{-1em}{-1em}{-.5em}{
    \begin{tikzpicture}[buchi]
      \small
      \begin{scope}[on above layer]
        \coordinate (init) at (0,0);
        \coordinate (center) at (1cm,0);
          \node[state,accept,right=0.2cm of init] (q0) {$\wcl{R}$};
          \node[state,accept,right=1cm of center] (q1) {$\wcl{\top{\conc}R}$};
          \node[state,accept,below=0.2cm of q0] (q2) {$\top$};
       \end{scope}
      \draw[move] (init) -- (q0);
      \draw[macc] (q0) to[out=5,in=175] node[above] {$\lowerlabel[-.5em]{\alpha}$} (q1);
      \draw[macc] (q1) to[out=185,in=-10] node[below] {$\top$} (q0);
      \draw[macc] (q0) to[out=200,in=180] node[left] {$\notA\alpha{\andA}\beta$} (q2);
      \draw[move] (q2) edge[loop right] node[right] {$\top$} (q2);
    \end{tikzpicture}
}
}
}
&
\multirow{3}{*}{
\textrm{
\mytikz{-2em}{-1em}{-1em}{-.5em}{
    \begin{tikzpicture}[buchi]
      \small
      \begin{scope}[on above layer]
        \coordinate (init) at (0,0);
        \coordinate (center) at (1cm,0);
          \node[state,right=0.2cm of init] (q0) {$\nwcl{R}$};
          \node[state,right=1cm of center] (q1) {$\nwcl{\top{\conc}R}$};
          \node[state,accept,below=0.2cm of q0] (q2) {$\top$};
       \end{scope}
      \draw[move] (init) -- (q0);
      \draw[macc] (q0) to[out=5,in=175] node[above] {$\lowerlabel[-.5em]{\alpha}$} (q1);
      \draw[macc] (q1) to[out=185,in=-10] node[below] {$\top$} (q0);
      \draw[macc] (q0) to[out=200,in=180] node[left] {$\notA\alpha{\andA}\notA\beta$} (q2);
      \draw[move] (q2) edge[loop right] node[right] {$\top$} (q2);
    \end{tikzpicture}
}
}
}
\\
\wcl{\top{\conc}R} & \ifthen{\top}{\wcl{R}}\feq\wcl{R}
& &
\\
\top(\feq\wcl{\eps}) & \top & &
\end{array}
\]
Then $a^\omega\models\wcl{R}$ because $\forall u\prec a^\omega:\exists x:ux\in\FLang{R}$.
Note also that
$\Lang{\Globally(\notA\beta)\land\wcl{R}}=\Lang{\ocl{(\alpha\notA\beta)}}$
where $\alpha$ holds in all even positions -- a classical example
not expressible in LTL \cite{Wolper83}.
\end{ex}

\begin{ex}
\label{ex:ocl}
With $\alpha$ and $\beta$ as in Example~\ref{ex:wcl} let
$\phi=\ocl{(\alpha\beta)}$. We get the following
transition terms. Many rewrite rules are being used here.
\[
\begin{array}{@{}l|l|l|l@{}}
q & \deriv{q}=\deriv{r\eimpl\Next\phi} & \One{r} & \NBA\;(\DBA)\;\M[\phi] \\\hline
q_0=\phi &
\deriv{\alpha\beta\eimpl\Next\phi} =
\deriv{\bot\land\Next\phi}\lor(\der{\alpha\beta}\eimpl\Next\phi) & \bot &
\multirow{7}{*}{
\textrm{
\mytikz{-1.5em}{-.7em}{-1em}{-1em}{
    \begin{tikzpicture}[buchi]
      \small
      \begin{scope}[on above layer]
        \coordinate (init) at (0,0);
        \coordinate (center) at (1cm,0);
          \node[state,accept,right=0.2cm of init] (q0) {$q_0$};
          \node[state,right=1cm of center] (q1) {$q_1$};   
      \end{scope}
      \draw[move] (init) -- (q0);
      \draw[macc] (q1) to[out=190,in=-10] node[below] {$\beta$} (q0);
      \draw[macc] (q0) to[out=10,in=170] node[above] {$\alpha$} (q1);
    \end{tikzpicture}
}
}
}
\\
&\feq (\ifthen{\alpha}{\beta}\eimpl\Next\phi) =
\ite{\alpha}{\beta\eimpl\Next\phi}{\bot\eimpl\Next\phi} & \\
&\feq  \ifthen{\alpha}{\beta\eimpl\Next\phi} &
\\
q_1=\beta\eimpl\Next\phi &
\deriv{\beta\land\Next\phi}\lor(\der{\beta}\eimpl\Next\phi) & \beta\\
&=\deriv{\beta}\land\deriv{\Next\phi}\lor(\ifthen{\beta}{\eps}\eimpl\Next\phi) & \\
&=(\ifthen{\beta}{\top}\land\phi)\lor\ite{\beta}{\eps\eimpl\Next\phi}{\bot\eimpl\Next\phi} & \\
&\feq \ifthen{\beta}{\phi} \lor \ite{\beta}{\bot}{\bot} \feq
\ifthen{\beta}{\phi} \lor \bot \feq \ifthen{\beta}{\phi} &
\end{array}
\]
Observe that $\ocl{(\alpha\beta)}$ is equivalent with
$\Globally(\alpha{\orA}\beta)\land\wcl{(\alpha\beta)\st{\conc}\notA(\alpha{\orA}\beta)}$
if $\notA(\alpha{\orA}\beta)$ is satisfiable.
\end{ex}

Theorem~\ref{thm:ELTL} describes the semantics of $\RLTL$
in terms of languages and derivatives. It
shows that the definition of derivatives correctly captures the
intended semantics. Observe that this theorem does not yet imply any formal
relationship with the corresponding $\ABA$ semantics that arises, that is discussed
in the next section.
For $L\subseteq\Ds$, let
\[
\MinSet{L} \eqdef \{u \in L\mid \not\exists x\prece u:x\in L\}
\]
i.e., $\MinSet{L}$ denotes the set of all words in $L$
that are \emph{minimal} in the sense that they have no proper prefixes in $L$.
In particular, $\Null{R}\IFF\MinSet{\FLang{R}}=\{\epsilon\}$.
Theorem~\ref{thm:ELTL} uses Lemma~\ref{lma:MinPrefix}.
\newcommand{\lemmaMinPrefixOne}{\textrm{\ref{lma:MinPrefix}(\ref{lma:MinPrefix:1})}}
\newcommand{\lemmaMinPrefixTwo}{\textrm{\ref{lma:MinPrefix}(\ref{lma:MinPrefix:2})}}
\newcommand{\lemmaMinPrefixThree}{\textrm{\ref{lma:MinPrefix}(\ref{lma:MinPrefix:3})}}
\begin{lma}
\label{lma:MinPrefix}
Let $R\in\ERE$, $w\in\Do$, and $a\in\D$. Then
\begin{enumerate}
\item 
\label{lma:MinPrefix:1}
${w \models\wcl{R} \IFF
(\exists u{\prece}w:u\in\MinSet{\FLang{R}}) \bor
(\forall u{\prec}w:\exists x:ux\in\FLang{R})}$
\item If $R$ is not nullable then
\label{lma:MinPrefix:2}
$\forall u\in\Ds: {au\in \MinSet{\FLang{R}} \IFF u \in \MinSet{\FLang{\leafof{\der{R}}{a}}}}$
\item If $R$ is not nullable then
\label{lma:MinPrefix:3}
$(\forall u\prec aw: \exists x:ux\in\FLang{R})
\IFF
(\forall u\prec w: \exists x:ux\in\FLang{\leafof{\der{R}}{a}})$
\end{enumerate}
\end{lma}
\begin{proof}
Using that $\exists\,u\prece w: u \in\FLang{R} \IFF 
\exists u\prece w: u\in\MinSet{\FLang{R}}$, and that
$au\in\FLang{R} \IFF u\in\FLang{\leafof{\der{R}}{a}}$.
In (\ref{lma:MinPrefix:3}) note that $ux\in\FLang{R}$ iff
$\Rest{u}x\in\FLang{\leafof{\der{R}}{\ith{u}{0}}}$ where $a=\ith{u}{0}$.
The case of $\exists x:ax\in\FLang{R}$
in the direction `$\Leftarrow$'
is implied in (\ref{lma:MinPrefix:3})
by $\leafof{\der{R}}{a}$ being alive.
\end{proof}
The following lemma is also used. While it is a well-known fact for
LTL here we show that it also holds in $\RLTL[\A]$ and is independent of $\D$.
\newcommand{\lemmaU}{\ref{lma:UR}(1)}
\newcommand{\lemmaR}{\ref{lma:UR}(2)}
\begin{lma}
\label{lma:UR}
$\forall \varphi,\psi\in\RLTL[\A]$:
1)
$\varphi\Until\psi \equiv \psi\lor (\varphi\land \Next(\varphi\Until\psi))$
and
2)
$\varphi\Release\psi \equiv \psi\land
(\varphi\lor \Next(\varphi\Release\psi))$
\end{lma}

\begin{thm}[Derivation]
\label{thm:ELTL}
$\forall\phi\in\RLTL[\A],a\in\D,w\in\Do:aw\models\phi \IFF w\models\leafof{\deriv{\phi}}{a}$.
\end{thm}
\begin{proof}
Fix $w\in\Do$ and $a\in\D$.
We prove the statement by 
by induction over $\phi$.

\noindent
\emph{Case $\alpha\in\PA$}:
$
aw\models\alpha \IFF a\in\den{\alpha} \IFF w\models\leafof{\ite{\alpha}{\top}{\bot}}{a}\IFF
w\models\leafof{\deriv{\alpha}}{a}
$.

\noindent
\emph{Case $\Next\varphi$}:
$
aw\models\Next\varphi \IFF
w\models\varphi \IFF w\models\leafof{\varphi}{a}
\IFF w\models\leafof{\deriv{\Next\varphi}}{a}
$.

\noindent
\emph{Case $\varphi\lor\psi$}:
\begin{myeq}
\begin{array}{@{}r@{\;}c@{\;}lc@{\;}l@{}}
aw\models\varphi \lor \psi &\IFF& aw\models\varphi \bor aw\models\psi \stackrel{\textrm{IH}}{\IFF} 
w\models \leafof{\deriv{\varphi}}{a}\lor\leafof{\deriv{\psi}}{a} \\
&\stackrel{\textrm{Lma}~\ref{lma:TT}}{\IFF}&
w\models \leafof{(\deriv{\varphi}\lor\deriv{\psi})}{a} 
\IFF
w \models \leafof{\deriv{\varphi\lor\psi}}{a}
\end{array}
\end{myeq}

\noindent
\emph{Case $\lnot\varphi$}:
\begin{myeq}
\begin{array}{@{}r@{\;}c@{\;}lc@{\;}l@{}}
aw\models\lnot\varphi &\IFF&
aw\not\models \varphi
\stackrel{\textrm{IH}}{\IFF}
w\not\models\leafof{\deriv{\varphi}}{a}
\IFF
w\models\lnot(\leafof{\deriv{\varphi}}{a}) \\
&\stackrel{\textrm{Lma}~\ref{lma:TT}}{\IFF}&
w\models\leafof{(\lnot\deriv{\varphi})}{a}
\IFF
w\models\leafof{\deriv{\lnot\varphi}}{a}
\end{array}
\end{myeq}

\noindent
\emph{Case $\phi=\varphi\Until\psi$}:
\begin{myeq}
\begin{array}{@{}r@{\;}c@{\;}l@{}}
aw\models\varphi\Until\psi &\stackrel{\textrm{Lma}~\ref{lma:UR}}{\IFF}&
aw\models\psi\lor (\varphi\land \Next(\varphi\Until\psi))
\IFF
aw\models\psi\bor (aw\models \varphi\band w\models\varphi\Until\psi)) \\
&\stackrel{\textrm{IH}}{\IFF}&
w\models\leafof{\deriv{\psi}}{a} \lor (\leafof{\deriv{\varphi}}{a} \land (\varphi\Until\psi))
\stackrel{\textrm{Lma}~\ref{lma:TT}}{\IFF}
w\models\leafof{(\deriv{\psi}\lor (\deriv{\varphi} \land (\varphi\Until\psi)))}{a} \\
&\IFF&
w\models\leafof{\deriv{\varphi\Until\psi}}{a}
\end{array}
\end{myeq}

\noindent
\emph{Case $R\eimpl\phi$}:
\[
\begin{array}{@{}r@{\;}c@{\;}l@{}}
aw\models R\eimpl\phi &\IFF& \exists u,v:aw=uv \band u\ith{v}{0}\in\FLang{R} \band v\models\phi \\
&\IFF& (a \in \FLang{R} \band  aw \models\phi) \bor 
\exists u,v: aw=auv \band au\ith{v}{0}\in\FLang{R} \band v\models\phi \\
&\IFF& (aw\models\One{R} \band aw \models\phi) \bor 
\exists u,v: w=uv \band u\ith{v}{0}\in\FLang{\leafof{\der{R}}{a}} \band v\models\phi \\
&\IFF& aw\models\One{R}\land\phi \bor w\models \leafof{\der{R}}{a}\eimpl\phi \\
&\stackrel{\textrm{IH}}{\IFF}&
w \models \leafof{\deriv{\One{R}\land\phi}}{a}  \bor w\models \leafof{\der{R}}{a}\eimpl\phi \\
&\IFF& w \models \leafof{\deriv{\One{R}\land\phi}}{a} \lor (\leafof{\der{R}}{a}\eimpl\phi) \\
&\stackrel{\textrm{Lma}~\ref{lma:TT}}{\IFF}& w \models \leafof{(\deriv{\One{R}\land\phi} \lor \der{R}\eimpl\phi)}{a} \\
&\IFF& w \models \leafof{\deriv{R\eimpl\phi}}{a}
\end{array}
\]

\noindent
\emph{Case $\wcl{R}$}: If $\Null{R}=\True$ then
\[
aw\models\wcl{R}\;\IFF\; aw\models\top\;\IFF\; w\models\leafof{\top}{a}
\;\IFF\; w\models\leafof{\deriv{\wcl{R}}}{a}.
\]
Let $\Null{R}=\False$ then
\[
\begin{array}{@{}r@{\;}c@{\;}l@{}}
aw\models \wcl{R}
&\stackrel{\textrm{Lma}~\lemmaMinPrefixOne}{\IFF}&
(\exists u\prece aw: u\in\MinSet{\FLang{R}}) \bor
(\forall u\prec aw: \exists x:ux\in\FLang{R})
\\
&\stackrel{\textrm{Lma}~\lemmaMinPrefixTwo,\lemmaMinPrefixThree}{\IFF}&
(\exists u\prece w: u\in\MinSet{\FLang{\leafof{\der{R}}{a}}})
\bor
(\forall u\prec w: \exists x:ux\in\FLang{\leafof{\der{R}}{a}})
\\
&\stackrel{\textrm{Lma}~\lemmaMinPrefixOne}{\IFF}& w\models \wcl{\leafof{\der{R}}{a}}
\stackrel{\textrm{Lma}~\ref{lma:TT}}{\IFF} w\models \leafof{\wcl{\der{R}}}{a}
\IFF w\models \leafof{\deriv{\wcl{R}}}{a}
\end{array}
\]

\noindent
\emph{Case $\ocl{R}$}: We use that
$a\in\FLang{R}\IFF \epsilon\in\FLang{\leafof{\der{R}}{a}}$. Recall
also that, for all $v\in\Do$, $v\models\Next\varphi \IFF \Rest{v}\models\varphi$.
In this case it is easier to start with $w\models \leafof{\deriv{\ocl{R}}}{a}$.
\[
\begin{array}{@{}r@{\;}c@{\;}l@{}}
w\models \leafof{\deriv{\ocl{R}}}{a} &\IFF&
w\models \leafof{\deriv{R\eimpl \Next\ocl{R}}}{a} \\
&\IFF& w\models\leafof{(\deriv{\One{R}\land\Next\ocl{R}}\lor\der{R}\eimpl\Next\ocl{R})}{a} \\
&\IFF& w\models\leafof{(\deriv{\One{R}}\land\ocl{R}\lor\der{R}\eimpl\Next\ocl{R})}{a} \\
&\IFF& w\models\leafof{(\ifthen{\One{R}}{\ocl{R}}\lor\der{R}\eimpl\Next\ocl{R})}{a} \\
&\stackrel{\textrm{Lma}~\ref{lma:TT}}{\IFF}&
a\in \FLang{R} \band w\models\ocl{R} \bor w\models \leafof{\der{R}}{a}\eimpl\Next\ocl{R} \\
&\IFF&
a\in \FLang{R} \band w\models\ocl{R} \bor
\exists u,v: w=uv\band u\ith{v}{0}\in\FLang{\leafof{\der{R}}{a}} \band v\models\Next\ocl{R} \\
&\IFF&
a\in \FLang{R} \band w\models\ocl{R} \bor
\exists u\neq\epsilon,v: w=uv\band u\in\FLang{\leafof{\der{R}}{a}} \band v\models\ocl{R} \\
&\IFF&
\exists u=\epsilon,v: w=uv\band u\in\FLang{\leafof{\der{R}}{a}} \band v\models\ocl{R} \bor \\
&& \exists u\neq\epsilon,v: w=uv\band u\in\FLang{\leafof{\der{R}}{a}} \band v\models\ocl{R} \\
&\IFF&
\exists u,v: w=uv\band u\in\FLang{\leafof{\der{R}}{a}} \band v\models\ocl{R}\\
&\IFF&
\exists u,v: w=uv\band au\in\FLang{R} \band v\in \ocl{(\FLang{R}\setminus\{\epsilon\})} \\
&\IFF&
aw \in (\FLang{R}\setminus\{\epsilon\}){\conc}\ocl{(\FLang{R}\setminus\{\epsilon\})}
\IFF
aw \in \ocl{(\FLang{R}\setminus\{\epsilon\})}
\IFF
aw \models \ocl{R}
\end{array}
\]
The other cases follow by de Morgan's laws and laws of duality in $\RLTL$.
The theorem follows by the induction principle.
\end{proof}

\subsection{$\omega$-Regularity Modulo Theories}
\label{sec:omegaRLTL}

Here we lift the classical concept of $\omega$-regular
languages \cite{Buchi60,McNaughton66} as the languages accepted by $\ABA$ so as to
be modulo $\A$, and show that $\RLTL[\A]$ captures $\omega$-regularity
modulo $\A$. 
We say that $L\subseteq\Do$ is \emph{$\omega_{\A}$-regular}
if $L=\Lang{\M}$ for some $\M$ that is an $\NBA[\A]$. We lift the
following result from \cite{Buchi60,McNaughton66} (see also
\cite[1.1.~Theorem]{ThomasAutomataOnInfiniteObjects})
as to be modulo $\A$.
Recall that $\Lang{\ocl{R}} = \ocl{(\FLang{R}{\setminus}\{\epsilon\})}$.
\begin{thm}[$\omega_{\A}$-regularity]
\label{thm:omega}
A language $L\subseteq\Do$ is $\omega_{\A}$-regular $\IFF$ there exist
$\{R_i\}_{i=1}^n$ and $\{S_i\}_{i=1}^n$ in $\RE[\A]$ such that
$L= \bigcup_{i=1}^n\FLang{R_i}{\conc}\Lang{\ocl{S_i}}$.
\end{thm}
\begin{proof}
By using Lemma~\ref{lma:mt} and the Theorem in \cite{McNaughton66}.
\end{proof}

Next we show that $\RLTLp[\A]$ captures $\omega_{\A}$-regularity.
We use the following \emph{stepping} lemma of the transition function
of any $\ABA$, the proof of which is based on the definition of accepting runs.
The lifting of $\tfs{u}{\phi}$ below for $\phi\in\BCp{Q}$ 
lets $\leafof{\tf(\varphi\diamond\psi)}{a}\eqdef \leafof{\tf(\varphi)}{a}\diamond\leafof{\tf(\psi)}{a}$.

\begin{lma}
\label{lma:tf}
$\forall \M\,{\in}\,\ABA[\A],
\phi\,{\in}\,\BCp{Q_{\M}},
u\,{\in}\,\Ds,w\,{\in}\,\Do:uw\,{\in}\,\Lang[\M]{\phi} \IFF w\,{\in}\,\Lang[\M]{\tfs[\M]{u}{\phi}}$.
\end{lma}

Recall the construction of the $\ABA[\A]$ $\M[\phi]$ in
Definition~\ref{def:M-RLTLp}.
Theorem~\ref{thm:ELTL-omega}(1)`$\Rightarrow$' is proved by using
Theorem~\ref{thm:omega}. Theorem~\ref{thm:ELTL-omega}(2) (from
which also (1)`$\Leftarrow$' follows) is proved by extending the
inductive proof in \cite{Vardi95LTL} to handle suffix implications
and also makes use of Lemma~\ref{lma:tf},
while closure formulas are additional base cases that are essentially
based on derivatives of $\ERE$.

\begin{thm}[$\omega_{\A}$-Regularity of {$\RLTLp[\A]$}]
\label{thm:ELTL-omega}\ 
\begin{enumerate}
\item $L\subseteq\Do$ is $\omega_{\A}$-regular $\IFF$ $\exists \phi\in\RLTLp[\A]:L = \Lang{\phi}$
\item $\forall\phi\in\RLTLp[\A]:\Lang{\phi} = \Lang{\M[\phi]}$
\end{enumerate}
\end{thm}

Although derivatives of $\ocl{R}$ work correctly in
Theorem~\ref{thm:ELTL}, and therefore also when considering
$\lnot\ocl{R}$ because the semantics is based on $\RLTL$,
\emph{the complemented derivatives arising from
$\lnot\deriv{\ocl{R}}$ would in general have incorrect
semantics as transitions of an $\ABA$}.

Moreover, sometimes $\omega$-closure $\ocl{R}$ can be
avoided through \emph{anchoring}: Let $\anchorpred$ be a new predicate called
an \emph{anchor}, not present in $\PA$, in the extended algebra $\A'$ with
$\D_{\A'}=\D\cup\{\anchorsymb\}$\footnote{If 
$\A$ is an SMT solver then $\anchorpred$ can be a fresh uninterpreted 
proposition, in which case $\anchorpred\mapsto\False$ when considering
any predicate $\alpha\in\PA$ and $\anchorpred\mapsto\True$ otherwise.}
such
that $\den[\A']{\anchorpred}=\{\anchorsymb\}$ and thus
$\den[\A']{\notA\anchorpred}=\D$.

\begin{ex}
\label{ex:aa}
Let $\alpha\in\PA$ be such that $\den{\alpha}=\{a\}$ and let
$\phi=\ocl{(\alpha{\conc}\alpha)}$. Then the following transition terms arise
similarly to $\phi$ in Example~\ref{ex:ocl} with $\phi$ as both the initial and the accepting state:
\[
\begin{array}{@{}l|l|l|l@{}}
q & \deriv{q}=\deriv{r\eimpl\Next\phi} & \One{r} & \DBA\;\M[\ocl{(\alpha{\conc}\alpha)}] \\\hline
\phi &
\deriv{\alpha\alpha\eimpl\Next\phi} \feq \ifthen{\alpha}{\alpha\eimpl\Next\phi} & \bot &
\multirow{2}{*}{
\textrm{
\mytikz{-1.5em}{-.7em}{-1em}{-.5em}{
    \begin{tikzpicture}[buchi]
      \small
      \begin{scope}[on above layer]
        \coordinate (init) at (0,0);
        \coordinate (center) at (1cm,0);
          \node[state,accept,right=0.2cm of init] (q0) {$\phi$};
          \node[state,right=1cm of center] (q2) {$\alpha\eimpl\Next\phi$};   
      \end{scope}
      \draw[move] (init) -- (q0);
      \draw[macc] (q2) to[out=175,in=10] node[above] {$\lowerlabel[-.6em]{\alpha}$} (q0);
      \draw[macc] (q0) to[out=-10,in=185] node[below] {$\raiselabel[-1.8em]{\alpha}$} (q2);
    \end{tikzpicture}
}
}
}
\\
\alpha\eimpl\Next\phi &
\deriv{\alpha\land\Next\phi}\lor(\der{\alpha}\eimpl\Next\phi) \feq \ifthen{\alpha}{\phi} & \alpha &
\end{array}
\]
If we were to complement the derivatives, we would obtain
the transitions:
\[
\begin{array}{@{}l|l|l@{}}
q & \deriv{q} & \textit{\textbf{Incorrect} $\DBA$}\\\hline
q_0=\lnot\phi &
\lnot\ifthen{\alpha}{\alpha\eimpl\Next\phi} = \ite{\alpha}{\lnot(\alpha\eimpl\Next\phi)}{\top}
&
\multirow{3}{*}{
\textrm{
\mytikz{-1.5em}{-1em}{-5em}{-1em}{
    \begin{tikzpicture}[buchi]
      \small
      \begin{scope}[on above layer]
        \coordinate (init) at (0,0);
        \coordinate (center) at (1.5cm,0);
          \node[state,right=0.2cm of init] (q0) {$q_0$};
          \node[state,accept,above=0.2cm of center] (q1) {$q_1$};
          \node[state,accept,right=1cm of center] (q2) {$\top$};   
      \end{scope}
      \draw[move] (init) -- (q0);
      \draw[move] (q0) -- node[above] {$\alpha$} (q1);
      \draw[move] (q0) -- node[below] {$\notA\alpha$} (q2);
      \draw[move] (q2) edge[loop above] node[right] {$\top$} (q2);
      \draw[macc] (q1) -- node[above] {$\notA\alpha$} (q2);
      \draw[macc] (q1) to[out=155,in=90] node[left] {$\alpha\;$} (q0);
    \end{tikzpicture}
}
}
}
\\
q_1=\lnot(\alpha\eimpl\Next\phi) & \lnot\ifthen{\alpha}{\phi} \feq \ite{\alpha}{\lnot\phi}{\top} &
\\
& & 
\end{array}
\]
Then
$q_1=\lnot(\alpha\eimpl\Next\phi)\equiv\alpha\uimpl\lnot\Next\phi$ would be
considered as an \emph{accepting} state and thus \emph{incorrectly} accept
$\ocl{a}$ because $\ocl{a}\models\ocl{(\alpha{\conc}\alpha)}$.  On the other
hand, for $\nwcl{(\alpha{\conc}\alpha)\st{\conc}\anchorpred}$
we get the transitions:
\[
\begin{array}{@{}l|l|l@{}}
s & \deriv{s} & \DBA\;\M[\nwcl{(\alpha{\conc}\alpha)\st{\conc}\anchorpred}]\\\hline
s_0=\nwcl{(\alpha\alpha)\st\anchorpred} &
\ite{\alpha}{\nwcl{\alpha(\alpha\alpha)\st\anchorpred}}{\ite{\anchorpred}{\nwcl{\eps}}{\nwcl{\bot}}} &
\multirow{4}{*}{
\textrm{
\mytikz{-1.5em}{-.5em}{-1em}{-1em}{
    \begin{tikzpicture}[buchi]
      \small
      \begin{scope}[on above layer]
        \coordinate (init) at (0,0);
        \coordinate (center) at (1cm,0);
          \node[state,right=0.2cm of init] (s0) {$s_0$};
          \node[state,right=2cm of center] (s1) {$s_1$};
          \node[state,below=0.4cm of s0] (s2) {$\bot$};
          \node[state,accept,below=0.4cm of s1] (top) {$\top$};   
       \end{scope}
      \draw[move] (init) -- (s0);
      \draw[macc] (s1) to[out=170,in=10] node[above] {$\lowerlabel[-.6em]{\alpha}$} (s0);
      \draw[macc] (s0) to[out=-10,in=190] node[above] {{$\lowerlabel[-.6em]{\alpha}$}} (s1);
      \draw[move] (s0) -- node[left] {$\notA\alpha{\andA}\anchorpred$} (s2);
      \draw[move] (top) edge[loop right] node[below] {$\top$} (top);
      \draw[macc] (s2)  edge[loop right] node[below] {$\top$} (s2);
      \draw[move] (s1) -- node[right] {$\notA\alpha$} (top);
      \draw[macc] (s0) to[out=-45,in=180] node[above] {$\notA\alpha{\andA}\notA\anchorpred$} (top);
    \end{tikzpicture}
}
}
}
\\
&\feq
\ite{\alpha}{\nwcl{\alpha(\alpha\alpha)\st\anchorpred}}{\ite{\anchorpred}{\bot}{\top}}&
\\
s_1=\nwcl{\alpha(\alpha\alpha)\st\anchorpred}
&
\ite{\alpha}{\nwcl{(\alpha\alpha)\st\anchorpred}}{\nwcl{\bot}} &\\
&\feq
\ite{\alpha}{\nwcl{(\alpha\alpha)\st\anchorpred}}{\top} &
\end{array}
\]
where we also show $\bot$ for once.
The conjunction $\Globally\notA\anchorpred\land\nwcl{(\alpha\alpha)\st{\conc}\anchorpred}$
restricts all transitions to $\D$:
\[
\begin{array}{@{}l|l|l|l@{}}
q & \deriv{q} & \DBA[\A']\;\M[\Globally\notA\anchorpred\land\nwcl{(\alpha\alpha)\st{\conc}\anchorpred}]
& \DBA[\A]\;\M[\lnot(\ocl{(\alpha\alpha)})] \\\hline
q_0=\Globally\notA\anchorpred\land s_0 &
\ite{\alpha}{\Globally\notA\anchorpred\land s_1}{
\ifthen{\notA\anchorpred}{\Globally\notA\anchorpred}} &
\multirow{3}{*}{
\textrm{
\mytikz{-2em}{-1em}{-5em}{-1em}{
    \begin{tikzpicture}[buchi]
      \small
      \begin{scope}[on above layer]
        \coordinate (init) at (0,0);
        \coordinate (center) at (1.5cm,0);
          \node[state,right=0.2cm of init] (q0) {$q_0$};
          \node[state,above=0.2cm of center] (q1) {$q_1$};
          \node[state,accept,right=1.5cm of center] (q2) {$q_2$};   
      \end{scope}
      \draw[move] (init) -- (q0);
      \draw[move] (q0) -- node[above] {$\alpha$} (q1);
      \draw[move] (q0) -- node[below] {$\notA\alpha{\andA}\notA\anchorpred$} (q2);
      \draw[move] (q2) edge[loop above] node[right] {$\notA\anchorpred$} (q2);
      \draw[macc] (q1) -- node[above] {$\notA\alpha{\andA}\notA\anchorpred$} (q2);
      \draw[macc] (q1) to[out=155,in=90] node[left] {$\alpha\;$} (q0);
    \end{tikzpicture}}
}
}
&
\multirow{3}{*}{
\textrm{
\mytikz{-2em}{-1em}{-5em}{-1em}{
    \begin{tikzpicture}[buchi]
      \small
      \begin{scope}[on above layer]
        \coordinate (init) at (0,0);
        \coordinate (center) at (1.5cm,0);
          \node[state,right=0.2cm of init] (q0) {$q_0$};
          \node[state,above=0.2cm of center] (q1) {$q_1$};
          \node[state,accept,right=1cm of center] (q2) {$\top$};   
      \end{scope}
      \draw[move] (init) -- (q0);
      \draw[move] (q0) -- node[above] {$\alpha$} (q1);
      \draw[move] (q0) -- node[below] {$\notA\alpha$} (q2);
      \draw[move] (q2) edge[loop above] node[right] {$\top$} (q2);
      \draw[macc] (q1) -- node[above] {$\notA\alpha$} (q2);
      \draw[macc] (q1) to[out=155,in=90] node[left] {$\alpha\;$} (q0);
    \end{tikzpicture}}
}
}
\\
q_1=\Globally\notA\anchorpred\land s_1
&
\ite{\alpha}{\Globally\notA\anchorpred\land s_0}{
\ifthen{\notA\anchorpred}{\Globally\notA\anchorpred}} & &
\\
q_2=\Globally\notA\anchorpred & \ifthen{\notA\anchorpred}{\Globally\notA\anchorpred} & &
\end{array}
\]
where $\notA\anchorpred$ is obsolete in $\A$ (becomes $\top$).
In particular observe that $\ocl{a}\notin\Lang{\M[\lnot(\ocl{(\alpha\alpha)})]}$.
\end{ex}

\section{Related Work}
\label{sec:related}

\paragraph*{Derivatives}

Transition regexes for extended regular
expressions \cite{StanfordVB21}, a symbolic generalization of
\emph{Brzozowski derivatives} \cite{Brz64}, is one of the inspirations
behind our work here. We view transition terms for LTL 
as a symbolic generalization of \emph{Vardi's derivatives for LTL} \cite{Vardi95LTL}.
Another source of inspiration in this context was \cite{DV21}
to lift the theory of automata modulo theories to $\omega$-languages.
LTL based \emph{monitoring} also uses derivatives \cite{SRA03,HR01},
although not in the generalized form of symbolic derivatives modulo
$\A$, but for a finite set of atomic propositions.
The relationship between the tableaux method for
LTL \cite{Wolper85} and applying the concept of
linear factors from \cite{Ant95} to LTL were
studied in \cite{Thiemann18}, implying a fundamental connection
between the two constructions.

\paragraph*{Alternation elimination}
Alternation elimination in regexes in \cite{StanfordVB21} through
symbolic regex derivatives amounts to incremental unfolding into NFAs
modulo $\A$ ($\A$ here is a \emph{Unicode character theory}) -- as a
symbolic unfolding of \emph{Antimirov derivatives} \cite{Ant95}.
The work in \cite{antimirovMosses95} uses Horn-equational reasoning
as an alternative to derivatives for reasoning about regular expressions that use
intersection and also allows infinite alphabets.
Symbolic regex derivatives are used in the core of the regex decision
procedure in Z3 \cite{BM08}.  The alternation elimination
algorithm \cite{MH84} is further studied in-depth \cite{BokerKR10}
with an improved lower bound 
for \emph{ordered $\ABA$s}.  An open question related
to \cite[Theorem~4]{BokerKR10}, proving state bound $2^{|\D|+n}$, is
its variant in the symbolic case when $\D$ can be \emph{infinite}.

Non-emptiness of $\ABA$ resulting from LTL is studied by Bloem, Ravi, and
Somenzi through classification into \emph{weak} and
\emph{terminal} cases where the problem is easier than in the general case \cite{BRS99}.

Fritz presents an LTL to $\NBA$ procedure \cite{Fritz03} based on
simulation relations for $\ABA$s, while starting with an $\ABA$ (based on
the Vardi construction) the procedure then computes simulation
relations on-the-fly, using delayed simulations and
$\epsilon$-transitions.  In the final phase the $\ABA$ is translated into
an $\NBA$ via \cite{MH84}.

The Gerth {et.al.} algorithm \cite{GerthPVW95} uses \emph{tableau} to
translate LTL into $\NBA$. The construction first constructs a
\emph{Generalized B\"uchi automaton} (GBA) in tableau using the expansion
rules of Vardi derivatives, that is subsequently transformed, via
a well-known automata-theoretic encoding into $\NBA$ that preserves that
\emph{all acceptance conditions} of the GBA are visited infinitely
often.  Tableau based techniques for LTL were initially studied by
Wolper \cite{Wolper81,Wolper83,Wolper85}.  A further extension of
the tableau based technique for LTL is introduced in \cite{Couvreur99}
using on-the-fly expansion of \emph{transition B\"uchi automata}.  The
key technique there is also rooted in what we call here Vardi
derivatives, that are called \emph{fundamental identities} of Boolean
variables in that context, that reflect how the variables for the
subformulas are created inductively.  
In regard to \emph{tableau modulo $\A$} in general,
it is an open and active research area, part of a general
effort to combine first-order deduction with modulo theories
with many open challenges \cite{Burel20}.

Vardi \cite[Theorem 14, Proof]{Vardi94} is the first LTL to
$\ABA$ construction defined in terms of a
step-wise unwinding essentially as derivatives. This
construction is not symbolic, as it uses the next element to directly
compute a Boolean combination of successor states.
\cite{TsayV21} give a full construction of LTL to
\emph{symbolic alternating co-B\"uchi automata.}  While aspects of the
construction are similar to the one in our work, by being based on a
symbolic representation of $\tfc(q,a)$ -- a key difference is in the
representation of $\tf(q)$ as a transition term providing a natural
separation of concerns between evaluation of transitions from their
target state formulas, that moreover works with \emph{any
effective Boolean algebra $\A$}.

Gastin and Oddoux \cite{GastinO01} modify Vardi's original
construction to produce very weak alternating co-B\"uchi automata
instead, that then are transformed into GBAs.
They take a step towards a symbolic representation by representing
transitions as relations, although treatment of the alphabet is still
non-symbolic. This representation allows them to develop
simplifications based on relational reasoning for eliminating implied
transitions and equivalent states \emph{on-the-fly}. In the context of
our work, since these simplifications operate directly on the
representation of the transition relation, the direct correspondence
between formulas and states is lost.

Wulf et al. \cite{WulfDMR08} define \emph{symbolic alternating B\"uchi
automata} (sABW) where transition relations are Boolean combinations
of literals and successor states. They further develop incremental
satisfiability and model checking methods using BDDs \cite{Bryant86} both as the
alphabet theory and to represent sets of states. They do not give the
construction from LTL to sABW, but instead refer to the presentations
in \cite{Vardi95LTL,GastinO01}.

Muller, Saoudi and Schupp \cite{MullerSS88} were the first to state
and prove the theorem that LTL can be translated to B\"uchi
automata. The proof, however, does not use an inductive unwinding of
LTL formula but composes automata from subformulas instead and is very
different from derivatives -- the construction uses weak alternating
automata over trees and a reduction \cite{MullerSS86} to B\"uchi
automata.

The ITE aspect of transition terms is not preserved in the above works
and the view of states as formulas, even if maintainable in some form
in a GBA, gets then lost in the translation to $\NBA$.

\paragraph*{Extensions of LTL with Background Theories}
\emph{LTL with constraints} \cite{Demri07} is a fragment of first-order
LTL, where LTL with \emph{Presburger arithmetic} \cite{Demri06} is a
typical extension of LTL.  While many decidable fragments exist, such
as CLTL$(\mathcal{D})$ \cite{Demri07}, constraint extensions in
general lead to undecidability.  Further extensions of LTL over
infinite domains are introduced in \cite{Kupferman_atva12} and
subsequently formalized with \emph{generalized register automata}
in \cite{Kupferman_atva13}.  In \cite{Kupferman_lpar18},
\emph{LTL with arithmetic or LTLA}
is studied further with a focus on its decidable
and undecidable fragments, where
the model-checking problem of
the existential fragment is shown to be in PSPACE and
the problem is studied also for hierarchical systems.

The synthesis problem of
\emph{nondeterministic looping word automata with arithmetic}
is studied in \cite{Kupferman_sofsem20}
by using games, where the languages are
also infinite words over an infinite alphabet,
in this case over rationals paired with a finite alphabet.

Recent work in \cite{GGG22} extends LTL with modulo theories over
\emph{finite} strings, the theory is also generally
\emph{undecidable}.

In the above works derivatives are not being used and the extensions
are in general orthogonal to what we propose here as $\LTL[\A]$ that maintains
decidability and algorithms modulo arbitrary effective Boolean algebra $\A$,
with an SMT solver as a prime example of
a general purpose effective Boolean algebra.
In some aspect $\LTL[\A]$ 
bridges a gap between classical LTL and certain decidable fragments of
theory specific extensions of LTL.

\paragraph{PSL and SPOT}
PSL \cite{PSL} is the \emph{IEEE 1850 Property Specification
Language}.  Its primary purpose is for formal specification of
concurrent systems.  SPOT \cite{spotRef} is a state-of-the-art verification tool that
supports a large subset of PSL.  $\RLTL$ covers a core subset of
PSL that is consistent with the semantics of suffix implications in
SPOT, but \emph{differs} in one important aspect in regard to \emph{weak
closure} (and its negation), namely that, for any nullable regex $R$,
$\wcl{R}\equiv\top$.  For a PSL formula $\wcl{R}$ the corresponding
weak closure in $\RLTL$ is $\wcl{R\inter\top\plus}$.  This
difference is driven by the semantics of derivatives that
effectively \emph{enforces} the law $\wcl{R}\equiv\top$
whenever $R$ is nullable, in order to to maintain an algebraically
well-behaved system of rules.  This also means that some basic rewrite
rules in SPOT (such as $\wcl{r\st}\equiv\wcl{r}$) are \emph{invalid}
for formulas in $\RLTL$.

$\RLTL$ also supports regex \emph{complement} that comes naturally
because of the built-in duality law of transition regexes:
$\compl\ite{\alpha}{f}{g}=\ite{\alpha}{\compl f}{\compl
g}$ \cite[Lemma~4.2]{StanfordVB21}. It is
difficult and highly impractical to support $\compl$ by other means,
because it would in general require \emph{determinization} of symbolic
finite automata that is avoided by the duality law that propagates
$\compl$ \emph{lazily}.

A key difference is that $\RLTL$ lifts (a core
subset of) PSL as to be modulo theories: while in PSL the atomic regexes
$\alpha$ are Boolean combinations of \emph{propositions}, in $\RLTL$
$\alpha$ can be any SMT formula.

Finally, further operators can be
supported by simply defining their
symbolic derivatives.  E.g., for regular
expression \emph{fusion} $(R\,{\fuse}\,S)$,
$\FLang{R\,{\fuse}\,S}= \{xay\mid x,y\in\Ds,a\in\D,xa\in\FLang{R},ay\in\FLang{S}\}$,
\[
\der{R\,{\fuse}\,S}\eqdef\ifthen{\One{R}}{\der{S}}\union(\der{R}\,{\fuse}\,S)
\qquad \Null{R\,{\fuse}\,S}\eqdef\False
\]
where again (\ref{eq:unary}) lifts $\lambda r.r\,{\fuse}\,S$ to transition regexes (extended with fusion).
Additional rewrite rules, such as $\eps\fuse S\feq \bot$, are also added.
Then, for example, for $\alpha,\beta\in\PA$ and $R=\alpha\st\fuse\beta\st$,
we get the following derivative, that we compare side-by-side with
the derivative of $S=\alpha\st{\conc}(\alpha{\andA}\beta){\conc}\beta\st$,
\[
\begin{array}{rcl}
\der{R} &\feq&
\ite{\alpha}{\ite{\beta}{(\beta\st\union R)}{R}}{\bot}
\\
\der{S} &\feq&
\ite{\alpha}{\ite{\alpha{\andA}\beta}{(\beta\st\union S)}{S}}{\bot}
\end{array}
\]
where $R$ and $S$ are clearly equivalent regexes, which is also reflected in their derivatives.

\section{Future Work}
\label{sec:future}
\paragraph*{Optimizations}
There are several layers of optimizations that can be
applied to the apporach presented here. 
All leaves in INF can be maintained as
\emph{antichains} based on the property of
maintaining all satisfiers in a minimal form -- recall that~(\ref{eq:MinSat})
is the antichain of $\DNF{\phi}$ with respect to $\subseteq$.
Use of antichains is a powerful general technique
also used in classical LTL \cite{WulfDMR08,FJR09}.

The construction of $\M[\wcl{R}]$ can be incremental just as in the
case of $\M[R]$ and can moreover eliminate dead states on-the-fly,
e.g., by using the algorithm in \cite{SV23}.  $\M[\wcl{R}]$ can also
be based on an \emph{NFA} of $R$ or symbolic generalization
of \emph{partial derivatives} \cite{Ant95}. NFAs require only minor
modifications for support in $\RLTL$ but can be more economical at
least in the case of $\RE[\A]$ ($\ERE[\A]$ without $\inter$ and
$\compl$). Use of NFAs avoids lazy \emph{determinization} that happens
implicitly in Lemma~\ref{lma:der} that is based on a symbolic
generalization of derivatives in \cite{Brz64}.  DFAs can be beneficial
when working with
\emph{negated weak closures} to avoid conjunctions
(and alternation elimination) that would otherwise arise
upon negation of transition terms.

For $\M[\lnot\ocl{R}]$ general $\NBA$ complementation of $\M[\ocl{R}]$
can be supported by first lifting the result that $\NBA$s are closed
under complement \cite{Buchi60} (see \cite[Theorem~4]{Kupf18}) as to
be modulo $\A$, through mintermization, and then develop symbolic
algorithms for complementing $\M[\ocl{R}]$ modulo $\A$.

At the \emph{transition} level: ITEs can potentially be
represented more efficiently by \emph{generic BDDs} \cite{DV17}, with
leaves as terminals, ordering the predicates in $\PA$ and maintaining
that order internally in Boolean operations,
or by utilizing \emph{Shannon expansions} \cite{Shannon49} in general when
working with ITEs.

At the \emph{automata} level: \emph{simulation}
algorithms \cite{GBS02,FritzWilke02,FritzWilke05,Fritz03,EWS05,MC13}
have been studied to reduce the state space, that have moreover been
improved with \emph{antichains} \cite{ACHMV10} and \emph{bisimulations
up to congruence} \cite{Pous13}.  Some of these algorithms have also
been lifted to \emph{modulo $\A$} in the case of
deterministic \cite{DV14,LICS16} and
nondeterministic \cite{DAntoniV17,HLSVV18} finite automata, where the
underlying technique is to work with \emph{symmetric differences of
predicates in $\A$ to avoid explicit iteration over the alphabet} as
in the classical Hopcroft DFA minimization \cite{Hop71} and
Paige-Tarjan relational coarsest partition \cite{PT87} algorithms --
both being essentially \emph{bisimulation} algorithms for DFAs and
NFAs, respectively.  It is therefore plausible that an analogous
treatment of predicates in $\A$ can be extended for simulation,
bisimulation, and further $\NBA$ algorithms studied in \cite{Pous21}
when lifted to modulo $\A$.

\paragraph*{The deterministic case}
Symbolic derivatives might give new insights for $\DBA$
algorithms \cite{Kurshan87,FinSch05,BaierKatoen08,BabiakKRS12,baarir.14.forte,KupfermanR10,tourneur.17.misc,DieMusWal15}
when lifted to modulo $\A$.

\paragraph*{Other extensions}
Symbolic derivatives can be developed
with laws that provide incremental/lazy unfolding of \emph{counters}, that
may be applicable to \emph{LTL with
counting} \cite{Laroussinie2010CountingL}.
The \emph{past-time} operator does not increase expressivity of
LTL \cite{GPSS80} -- it is nevertheless practically very
useful \cite{LPZ85,LS02,KP95,KPV12}.  Symbolic derivatives
are inherently ``forward looking'',
so understanding how to efficiently handle past-time with them is intriguing.

Finally, a larger open problem is if the derivative based view can also be
developed modulo $\A$ for \emph{Computation Tree Logic} (CTL) or even
CTL* (a hybrid of LTL and CTL) \cite{Piterman2018}.  \emph{Infinite trees} rather than
words are the core elements in the semantics in this case.
Derivatives for finite trees have recently been studied in \cite{SLD21}.

\section{Conclusion}
\label{sec:conclude}

We have shown how symbolic transition terms and derivatives can be used to
define a realizable symbolic semantics for 
(alternating) B\"uchi automata and linear temporal logic (LTL).
The semantics is parameterized by an effective Boolean algebra
for the base alphabetic domain, which enables it to apply to
$\omega$-languages and infinite alphabets in an algebraically
well-defined and precise manner. This framework allows syntactic
rewrite rules for LTL extended with regular expressions (RLTL)
to be applied \emph{on-the-fly} during alternation
elimination, where they simultaneously respect the semantics of RLTL
formulas and their alternating B\"uchi automata. Similarly to the new
alternation elimination algorithm, we believe there is a rich
landscape of further optimizations and
algorithms yet to be discovered.


\bibliographystyle{ACM-Reference-Format}

\newpage

\appendix

\section{Proofs}


\begin{proof}[Proof of Corollary~\ref{cor:Product}]
Let $\M=\N[1]\land\N[2]=(\A,Q,\Qi,\tf,F)$ as
the $\ABA$ in the algorithm and let $\N=\AElim{\M}$.
We use that all leaves of $\tf_1$ and $\tf_2$ are disjunctions of states, i.e.,
in INF, all leaves are sets whose members are singleton sets.
Initially $S_0 \leftarrow \Fin{\DNF{\Qi[1]\inter\Qi[2]}}$.
Therefore $S_0$ consists of states of the form
$\Fin{\setset{q^0_1,q^0_2}}=\Pair{\set{q^0_1,q^0_2}\setminus F}{\set{q^0_1,q^0_2}\cap F}$,
where $q^0_1\in \Qi[1]$ and $q^0_2\in \Qi[2]$.
The invariant $U\cap F=\emptyset$ holds.
We now consider $q=\Pair{U}{V}$ in the algorithm where $f$ is computed and
show that all leaves of $f$ are sets of states of one of the four cases with
$q_1\in Q_1$ and $q_2\in Q_2$:
\[
1)\, q {=}\Pair{\set{q_1,q_2}}{\emptyset} \bor
2)\, q {=}\Pair{\emptyset}{\set{q_1,q_2}} \bor
3)\, q {=}\Pair{\set{q_1}}{\set{q_2}} \bor
4)\, q {=} \Pair{\set{q_2}}{\set{q_1}}.
\]

\noindent
\emph{Cases 1 and 2}:
$f = \tfINF(\set{q_1,q_2})\aprod\set{\emptyset} =
g\aprod\set{\emptyset}$ where
$g=\INF{\tf_1(q_1)\land\tf_2(q_2)}$ 
and so in $g$ each leaf $\bvarphi$ is a set of sets
of the form $X = \{p_1,p_2\}$ for some $p_1\in Q_1$ and $p_2\in Q_2$.
Then $\Pair{X\setminus F}{\emptyset\cup(X\cap F)}$ clearly corresonds to one of the four cases.

\noindent
\emph{Cases 3 and 4}: Let $\{i,j\}=\{1,2\}$. Then
$f= \tfINF(q_i)\aprod\tfINF(q_j)= \INF{\tf_i(q_i)}\aprod\INF{\tf_j(q_j)}$
where each member $X$ of a leaf in $\INF{\tf_i(q_i)}$ is a singleton $\set{p_i}$ with $p_i\in Q_i$,
and each member $Y$ of a leaf in $\INF{\tf_j(q_j)}$ is a singleton $\set{p_j}$ with $p_j\in Q_j$.
It follows that $\Pair{X{\setminus} F}{Y\cup(X\cap F)}$ belongs to one of the cases 2, 3 or 4,
because if $p_i\in F_i$ then $X{\setminus}F=\emptyset$ and $X\cap F = X$ else
$X{\setminus} F= X$ and $X\cap F = \emptyset$.

In each case the invariant $(X{\setminus}F)\cap F=\emptyset$ is clearly maintained.
\end{proof}


\begin{proof}[Proof of Theorem~\ref{thm:FA}]
Let $I$ be a
finite subset of $\PA$.  In the following, it is convenient to
identify each minterm $\alpha\in\Minterms{I}$ by the unique subset
$I_{\alpha}$ of $I$ such that $\alpha\equiv {{\textrm{\Large\&}_{\gamma\in I_\alpha}}
\gamma\andA{\textrm{\Large\&}_{\gamma\in I\setminus I_\alpha}}\notA\gamma}$.
Let $\CondOf{\phi}$ denote the set of all predicates in $\PA$
that occur in $\phi\in\LTL[\A]$.

We first translate $\phi$ into classical LTL as follows.
Let $I=\CondOf{\phi}$,
and for each $\iota\in I$ let $p_\iota$ be a fresh proposition in
$P=\{p_\iota\}_{\iota\in I}$. Let $\Sigma=\Minterms{I}$.  Map each
minterm $\alpha\in \Sigma$ uniquely to $\pset{\alpha}=\{p_\iota\in
P\mid \iota\in I_\alpha\}$.
Thus, $p_\iota\in \pset{\alpha}$ iff $\iota$ is a positive literal of $\alpha$.

Let $\pconj{\alpha}$ denote the conjunction
$\bigwedge_{p\in\pset{\alpha}}p\land\bigwedge_{p\in
P\setminus\pset{\alpha}}\lnot p$ that encodes the minterm $\alpha$
propositionally.
Now, for all
$\iota\in I$, let
$\pform{\iota}= \bigvee_{\alpha\in \Sigma,\alpha\IMP\iota}\pconj{\alpha}$
denote the propositional encoding of $\iota$,
where $\alpha\IMP\iota$ means that $\den{\alpha}\subseteq\den{\iota}$
that is equivalent to $\alpha{\andA}\iota\nequiv\bot$ by definition of minterms.
Observe that
$\pform{\iota}=\bot$ precisely when $\iota$ is \emph{unsatisfiable} in $\A$.

For clarity let $\models_P$ stand for the classical semantics of LTL over $P$,
while $\models$ is defined in (\ref{eq:M1}--\ref{eq:M7}).
As earlier, 
for all $a\in\D$ let $\mt{a}\in\Sigma$ denote the minterm such that $a\in\den{\mt{a}}$.
It follows that
\[
\pset{\mt{a}} \models_P
\pform{\iota}\;\IFF\; (\mt{a}\IMP\iota)
\;\IFF\; a\in\den{\iota} \;\IFF\; a\models\iota
\]
For $a\in\D$ and $w\in\Do$, 
lift $\mt{a}$ to $\mt{w}\in\Sigma^\omega$
and  lift
$\pset{\mt{\alpha}}$ to $\pset{\mt{w}}\in(2^P)^\omega$.
Finally, let $\pform{\phi}$ be the conversion of $\phi$ to classical LTL
by replacing each predicate $\iota\in I$ in $\phi$ by $\pform{\iota}$.
Then
\[
w \models \phi \;\IFF\; \pset{\mt{w}}\models_P\pform{\phi} \;\stackrel{\textrm{V95}}{\IFF}\;
\pset{\mt{w}} \in \Lang{\M[\pform{\phi}]}
\;\stackrel{(\star)}{\IFF}\;
\mt{w}\in \Lang{\mt{\M[\phi]}}
\;\stackrel{\textrm{Lma~\ref{lma:mt}}}{\IFF}\;
w\in \Lang{\M[\phi]}
\]
where V95 is \cite[Theorem~24]{Kupf18}; $(\star)$ follows on one
hand from renaming of the alphabet symbols, each symbol
$\pset{\alpha}\in 2^P$ is uniformly renamed to the symbol
$\alpha\in\Sigma$; on the other hand, $\mt{\M[\phi]}$ introduces
$\Sigma$ through mintermization of $I$ where the semantics is
preserved modulo renaming by definition of $\pform{\iota}$
that precisely encodes 
$|_{\alpha\in\Sigma,\alpha\IMP\iota}\alpha$
that is the mintermized equivalent form of $\iota$ in $\A$ as
the disjunction of all the minterms contained in $\iota$.
\end{proof}
The construction of $\pform{\iota}$ above can also use \emph{bit-blasting}
where minterms are directly encoded as \emph{independent} bits by propositions
$\{p_\alpha\}_{\alpha\in\Sigma}$. This alternative encoding
may be of practical interest is an actual conversion to the classical setting.

The following example illustrates the constructions used in the proof
of Theorem~\ref{thm:FA}.
\begin{ex}
Let $\A$ be a Boolean algebra of predicates over integers
and $\phi=(x{<}1)\Release(0{<}x)$ with $x$ of type \emph{integer}.
Then $I=\{x{<}1,0{<}x\}$, $P=\{p_{x{<}1},p_{0{<}x}\}$ 
and
$\Sigma=\{x{<}1{\andA}\notA0{<}x, 0{<}x{\andA}\notA x{<}1\}$.
The formulas $\pform{\iota}$ for $\iota\in I$ are
$\pform{x{<}1}=p_{x{<}1}{\land}\lnot p_{0{<}x}$ and
$\pform{0{<}x}=p_{0{<}x}{\land}\lnot p_{x{<}1}$.
And 
thus $\pform{\phi}=(p_{x{<}1}{\land}\lnot p_{0{<}x})\Release (p_{0{<}x}{\land}\lnot p_{x{<}1})$.

If we instead consider $x$ as a \emph{real} then
$\Sigma$ also includes $0{<}x{\andA}x{<}1$ while
$\notA 0{<}x{\andA}\notA x{<}1$ is still unsatisfiable.
Then
$\pform{x{<}1}=p_{x{<}1}{\land}\lnot p_{0{<}x}\lor p_{x{<}1}{\land}p_{0{<}x}$ that simplifies
to $\pform{x{<}1}=p_{x{<}1}$, and similarly we get that  $\pform{0{<}x}=p_{0{<}x}$,
in which case $\pform{\phi}=p_{x{<}1}\Release p_{0{<}x}$.
\end{ex}


\begin{proof}[Proof of Lemma~\ref{lma:UR}]
We show (1), the proof of (2) is analogous.

\noindent
\emph{Direction $(\Rightarrow)$}: Let $w\models \varphi\Until\psi$.
Fix $j\in\Nat$ such that (\ref{eq:M6}) holds.
If $j=0$ then $\Rest[0]{w} = w \models \psi$ and we are done.
If $j > 0$ then, for all $i<j$, $\Rest[i]{w}\models\varphi$ and
$\Rest[j]{w}=\Rest[j-1]{(\Rest{w})}\models\psi$. In particular, $\Rest[0]{w}=w\models\varphi$ and
there exists $k=j-1\in\Nat$ such that
$\Rest[k]{(\Rest{w})}\models\psi$ and
for all $i<k$, $\Rest[i]{(\Rest{w})}\models\varphi$. So, by (\ref{eq:M6}),
$\Rest{w}\models\varphi\Until\psi$.

\noindent
\emph{Direction $(\Leftarrow)$}: If $w\models\psi$ then  $w\models \varphi\Until\psi$
follows immediately from (\ref{eq:M6}).
If $w\models\varphi$ and $\Rest{w}\models\varphi\Until\psi$ then
there exists $k\in\Nat$ such that, by (\ref{eq:M6}),
$\Rest[k]{(\Rest{w})}\models\psi$ and
for all $i<k$, $\Rest[i]{(\Rest{w})}\models\varphi$.
It follows that for $j=k+1$, $\Rest[j]{w}\models\psi$
and
for all $i<j$, $\Rest[i]{w}\models\varphi$.
So  $w\models \varphi\Until\psi$ by (\ref{eq:M6}).
\end{proof}


\begin{proof}[Proof of Lemma~\ref{lma:tf}]
Fix $\M=(\A,Q,\Qi,\tf,F)$ and $w\in\Do$.
We prove the following statement by induction over the length of $u\in\Ds$:
\[
\forall
\phi\,{\in}\,\BCp{Q}: uw\,{\in}\,\Lang[\M]{\phi} \IFF w\,{\in}\,\Lang[\M]{\tfs{u}{\phi}}
\]
The empty word case holds by definition because $\tfs{\epsilon}{\phi}=\phi$.
For the nonempty word case we first prove the following statement:
\begin{equation}
\label{eq:tf}
\forall q\in Q,av\in\Do:av\,{\in}\,\Lang[\M]{q} \IFF v\,{\in}\,\Lang[\M]{\leafof{\tf(q)}{a}}
\end{equation}
\textit{Direction $\Rightarrow$ of (\ref{eq:tf})}:
Let $av\in \Lang[\M]{q}$.  So there exists an accepting
run $\tau$ for $av$ from $q$.
This implies that that there exists
$X\in\leafof{\tfINF(q)}{a}$ such that, for each $p\in X$,
$\tau$ has has an immediate subtree $\tau_p$ accepting a run for
$v$ from $p$.  In other words,
\[
v\in\bigcap_{p\in X}\Lang[\M]{p}=
\Lang[\M]{\bigwedge X}\subseteq\Lang[\M]{\leafof{\tfINF(q)}{a}}=\Lang[\M]{\leafof{\tf(q)}{a}}
\]
\textit{Direction $\Leftarrow$ of (\ref{eq:tf})}: Let $v\in \Lang[\M]{\leafof{\tf(q)}{a}}$.
So $v\in\Lang[\M]{\leafof{\tfINF(q)}{a}}$ and therefore
there exists $X\in\leafof{\tfINF(q)}{a}$ such that
$v\in\Lang[\M]{\bigwedge X}$. So there exist accepting
runs $\tau_p$ for $v$ from $p$ for each $p\in X$.
Let $\tau$ be the infinite tree with root labelled by $q$
having $\tau_p$, for $p\in X$, as its immediate subtrees.
Then $\tau$ is an accepting run for $av$ from $q$ and thus  $av\in \Lang[\M]{q}$.

\noindent
Next we lift (\ref{eq:tf}) to $\BCp{Q}$ where $\tf$ is lifted similarly (as with derivatives) and
use Lemma~\ref{lma:TT}.
\begin{equation}
\label{eq:tf2}
\forall \phi\in\BCp{Q},av\in\Do:av\,{\in}\,\Lang[\M]{\phi} \IFF v\,{\in}\,\Lang[\M]{\leafof{\tf(\phi)}{a}}
\end{equation}

\noindent
We now complete the induction case of the main statement.
\[
auw \in \Lang[\M]{\phi} 
\stackrel{(\ref{eq:tf2})}{\IFF} uw \in \Lang[\M]{\leafof{\tf(\phi)}{a}}
\stackrel{\textrm{IH}}{\IFF} w \in \Lang[\M]{\tfs{u}{\leafof{\tf(\phi)}{a}}} 
\IFF 
w \in \Lang[\M]{\tfs{au}{\phi}} 
\]
The lemma follows by the induction principle.
\end{proof}


\begin{proof}[Proof of Theorem~\ref{thm:ELTL-omega}] We first prove (1)`$\Rightarrow$',
and then (2) which imples (1)`$\Leftarrow$'.

\paragraph*{Proof of `$\Rightarrow$' of (1)}
Let $L=\Lang{\M}$ for some $\M\in\ABA[\A]$ and
let $R_i,S_i\in\RE[\A]$, for $1\leq i\leq n$ be given by Theorem~\ref{thm:omega}.
Let
\[
\begin{array}{l}
\phi=\bigvee_{i=1}^n\phi_i\;\textrm{where}\;
\phi_i=\ITEBrace{\Null{R_i}}{\ocl{S_i}\lor (R_i\eimpl\Next\ocl{S_i})}{R_i\eimpl\Next\ocl{S_i}}
\end{array}
\]
Fix $\phi_i$ and assume that $R_i$ is nullable. Then
$w \models \phi_i$ iff $w\models \ocl{S_i}$ 
or there exists $u\in\Ds$ and $v\in\Do$ such that $w=uv$ and
$u\ith{v}{0}\in\FLang{R_i}$ and $v\models\Next\ocl{S_i}$
iff there is $u\in\Dp$ and $v\in\Do$ such that
$u\in\FLang{R_i}$ and $v\models\ocl{S_i}$.
So $w\models\phi_i$ iff
$w\in\FLang{R_i}{\conc}\Lang{\ocl{S_i}}$.
The case when $R_i$ is not nullable is analogous.
It follows that $L=\Lang{\phi}$.

\paragraph*{Proof of (2)}
The proof is by extending the
original proof in \cite{Vardi95LTL} with the additional constructs.
We do not repeat the proof of the $\LTL$ fragment for the
modulo $\A$ case because it would repeat very similar arguments.
Let $\M=\M[\phi]$.

In the case of $\wcl{R}$ we get the $\DBA$ $\M[\wcl{R}]$ by fixpoint construction of derivatives
whose \emph{accepting} states are all the states $\wcl{r}$ such that $r$ is alive.
In the case of $\ocl{R}$ we get an $\NBA$ $\M[\ocl{R}]$ similarly but whose \emph{only accepting}
state is $\ocl{R}$. 
For \emph{negated weak closure} $\nwcl{R}$ the construction of  $\M[\nwcl{R}]$
is analogous with $\M[\wcl{R}]$ except that in $\M[\nwcl{R}]$
a state $\nwcl{r}$ is accepting iff $r$ is dead.
It follows in all those cases that $\Lang{\phi} = \Lang{\M[\phi]}$.

The main cases are $\eimpl$ and $\uimpl$ because they involve induction over
$\RLTLp$ formulas by building on the proof in \cite{Vardi95LTL}.
By construction of $\M$ we have that $\deriv{q}\feq\tf_{\M}(q)$ and we let $\tf=\tf_{\M}$ below.
We also use that
$\forall R\in\ERE[\A], a\in\D,u\in\Ds:a\in\den{\One{\ders{u}{R}}}\IFF ua\in\FLang{R}$.

Let $w\in\Do$. 
Let $\Prefix{i}{w}$
denote the finite prefix of $w$ up to (including) the element $\ith{w}{i}$.
\[
\begin{array}{@{}r@{\;}c@{\;}l@{}}
w\in\Lang[\M]{R{\eimpl}\phi} &\stackrel{(\ref{eq:eimpl-deriv}),\textrm{Lma~\ref{lma:tf}}}{\IFF}&
\Rest{w}\in\Lang[\M]{\leafof{(\tf(\One{R}\land\phi)\lor(\der{R}{\eimpl}\phi))}{\ith{w}{0}}}\\
&\stackrel{\textrm{Lma~\ref{lma:TT}}}{\IFF}&
(\ith{w}{0}{\in}\FLang{R}\band \Rest{w}{\in}\Lang[\M]{\tfs{\ith{w}{0}}{\phi}}) \bor
\Rest{w}{\in}\Lang[\M]{\ders{\ith{w}{0}}{R}{\eimpl}\phi} \\
&\stackrel{\textrm{Lma~\ref{lma:tf}}}{\IFF}&
(\ith{w}{0}{\in}\FLang{R}\band w\in\Lang[\M]{\phi}) \bor \Rest{w}\in\Lang[\M]{\ders{\ith{w}{0}}{R}{\eimpl}\phi} \\
&\stackrel{(\textrm{Lma~\ref{lma:tf}})^k}{\IFF}&
\bigvee_{i<k}(\Prefix{i}{w}\in\FLang{R} \band \Rest[i]{w}\in\Lang[\M]{\phi}) \\
&&
\bor \Rest[k]{w}\in\Lang[\M]{\ders{\Prefix{k{-}1}{w}}{R}\eimpl\phi} \\
&\stackrel{\textrm{(e)}}{\IFF}&
\exists i\geq 0: \Prefix{i}{w}\in\FLang{R} \band \Rest[i]{w}\in\Lang[\M]{\phi} \\
&\stackrel{\textrm{IH}}{\IFF}&
\exists i\geq 0: \Prefix{i}{w}\in\FLang{R} \band \Rest[i]{w}\models\phi \\
&\stackrel{(\ref{eq:eimpl})}{\IFF}&
w \models R{\eimpl}\phi
\end{array}
\]
In $(\textrm{Lma~\ref{lma:tf}})^k$ we repeat the previous unfolding a finite
number $k$ times.  In (e) we use the key property that all
states $r\eimpl\phi$ in $\M$ are \emph{rejecting}
and can therefore only be visited \emph{finitely} often in any accepting
run for $w$. Thus,
the disjunct with $\Rest[k]{w}$ acts as $\bot$ for large enough $k$. 
\[
\begin{array}{@{}r@{\;}c@{\;}l@{}}
w\in\Lang[\M]{R{\uimpl}\phi} &\stackrel{(\ref{eq:uimpl-deriv}),\textrm{Lma~\ref{lma:tf}}}{\IFF}&
\Rest{w}\in\Lang[\M]{\leafof{(\tf(\One{R}\limplies\phi)\land(\der{R}{\uimpl}\phi))}{\ith{w}{0}}}\\
&\stackrel{\textrm{Lma~\ref{lma:TT}}}{\IFF}&
({\ith{w}{0}\,{\in}\,\FLang{R}}\Rightarrow {\Rest{w}\,{\in}\,\Lang[\M]{\tfs{\ith{w}{0}}{\phi}}}) \band
\Rest{w}\,{\in}\,\Lang[\M]{\ders{\ith{w}{0}}{R}{\uimpl}\phi} \\
&\stackrel{\textrm{Lma~\ref{lma:tf}}}{\IFF}&
(\ith{w}{0}\,{\in}\,\FLang{R}\Rightarrow w\,{\in}\,\Lang[\M]{\phi})
\band \Rest{w}\,{\in}\,\Lang[\M]{\ders{\ith{w}{0}}{R}{\uimpl}\phi} \\
&\stackrel{(\textrm{Lma~\ref{lma:tf}})^m}{\IFF}&
\bigwedge_{i<m}(\Prefix{i}{w}\,{\in}\,\FLang{R} \Rightarrow \Rest[i]{w}\,{\in}\,\Lang[\M]{\phi}) \\
&&
\band \Rest[m]{w}\in\Lang[\M]{\ders{\Prefix{m{-}1}{w}}{R}\uimpl\phi} \\
&\stackrel{(\textrm{u})}{\IFF}&
\forall i\geq 0: \Prefix{i}{w}\in\FLang{R} \Rightarrow \Rest[i]{w}\in\Lang[\M]{\phi} \\
&\stackrel{\textrm{IH}}{\IFF}&
\forall i\geq 0: \Prefix{i}{w}\in\FLang{R} \Rightarrow \Rest[i]{w}\models\phi \\
&\stackrel{(\ref{eq:uimpl})}{\IFF}&
w \models R{\uimpl}\phi
\end{array}
\]
In $(\textrm{Lma~\ref{lma:tf}})^m$ we repeat the unfolding any finite
number $m$ times.  In (u) we use the key property that all
states $r\uimpl\phi$ in $\M$ are \emph{accepting}
and  do not interfere with the first conjunct. The second
conjunct with $\Rest[m]{w}$ acts as $\top$ in the limit.

The statement (2) follows by the induction principle
and the theorem follows.
\end{proof}

\end{document}